\documentclass{statsoc}

\usepackage[a4paper]{geometry}
\usepackage{graphicx}
\usepackage[textwidth=8em,textsize=small]{todonotes}

\usepackage{xcolor}
\usepackage{natbib}
\usepackage{bm, commath}
\usepackage{hyperref}
\usepackage{amsmath, amsfonts}
\DeclareMathOperator*{\argmax}{arg\,max}
\usepackage{float}

\usepackage{etoolbox}
\usepackage{changepage}

\usepackage{comment}
\usepackage{hyphenat}
\usepackage{multirow}

\newtheorem{definition}{Definition}

\newtheorem{proposition}{Proposition}

\title[Case Description Information and Sensitivity Analysis]{Using Case Description Information to Reduce Sensitivity to Bias for the Attributable Fraction Among the Exposed}
\author[Kan Chen ]{Kan Chen \thanks{Email: {\tt kanchen@hsph.harvard.edu}} }
\address{Harvard University, Boston, Massachusetts, U.S.A.}
\author{Jing Cheng \thanks{Email: {\tt 	Jing.Cheng@ucsf.edu}} }
\address{University of California San Francisco, San Francisco, California, U.S.A.}
\author{M.Elizabeth Halloran \thanks{Email: \tt betz@fredhutch.org} }
\address{Fred Hutchinson Cancer Center, Seattle, Washington, U.S.A.}
\author{Dylan S. Small \thanks{Email: {\tt dsmall@wharton.upenn.edu}}}
\address{University of Pennsylvania, Philadelphia, Pennsylvania, U.S.A.}

\begin{document}
\begin{abstract}
The attributable fraction among the exposed (\textbf{AF}$_e$) is the proportion of disease cases among the exposed that could be avoided by eliminating the exposure. In this paper, we propose a new approach to reduce sensitivity to hidden bias for conducting statistical inference on the \textbf{AF}$_e$ by leveraging case description information such as subtype of cancer. The proposed method is examined through an asymptotic tool, design sensitivity, simulation studies, and case studies of alcohol consumption and the risk of postmenopausal invasive breast cancer utilizing information on the subtype of cancer using data from the Women's Health Initiative Observational Study allowing the possibility that leveraging case definition information may introduce selection bias through an additional sensitivity parameter.
\end{abstract}


\section{Introduction}
\label{sec: intro}

\subsection{Background: The Attributable Fraction Among the Exposed}

 We consider performing statistical inference for the attributable fraction among the exposed (\textbf{AF}$_e$). The attributable fraction among the exposed, also known as the attributable risk (\textbf{AR}) \citep{levin1953occurrence, walter1975distribution, uter1999concept}, is the proportion of cases in the exposed group that are attributable to the risk factor. It represents the proportion of disease cases that could be avoided if this risk factor were eliminated from a population. To assess the public health importance of controlling a selected risk factor, it is common to use \textbf{AF}$_e$ as an epidemiological measure for determining the severity of a harmful exposure among those who have been exposed.


Over the past few decades, statistical inference on the attributable fraction among the exposed has mostly focused on the inference of $\textbf{AF}_e$ under the assumption of \emph{no unmeasured confounding}. Examples of methods for inference that have been implemented include using Wald’s statistic \citep{kuritz1987attributable, kuritz1988attributable}, logistic regression \citep{drescher1991attributable}, logit transformation \citep{leung1981comparisons}, and maximum likelihood estimation \citep{greenland1993maximum}. In observational studies, because no one can guarantee the results are unaffected by unobserved covariates, the inference for $\textbf{AF}_e$ requires a sensitivity analysis that assesses the sensitivity of causal conclusions to the violation of the \emph{no unmeasured confounder} assumption.

\subsection{Sensitivity Analysis: No Unmeasured Confounding Assumption}


To conduct a sensitivity analysis in inference of the $\textbf{AF}_e$ in the presence of unobserved covariates, throughout this paper, we adapt Rosenbaum's odds ratio framework \citep{rosenbaum1987sensitivity}. \citet{rosenbaum1987sensitivity} derives bounds on the $P$-value of a test for whether there is a treatment effect given an upper limit on the magnitude of bias, say $\Gamma$, on the odds ratio of the probability of treatment assignment for a set of units matched on observed covariates. When $\Gamma = 1$, the \emph{no unmeasured confounder} assumption holds.

No unmeasured confounding is somehow the most favorable situation; however, a researcher, in practice, cannot know if he/she is in such a position. Conducting a robust inference in the sense that causal conclusions are insensitive to small or moderate hidden bias measured by $\Gamma$ has increasingly attracted attention from researchers and practitioners alike (\citet{stuart2013commentary}).  



\subsection{Our Contribution: Using Case Description Information to Reduce Sensitivity to Hidden Bias, and Sensitivity Analysis for Selection Bias}

In this paper, we present a novel approach to reducing sensitivity to hidden bias in inference for the $\textbf{AF}_e$ by leveraging case description information. Case description information refers to information that describes types of cases. As a motivating example, we consider inference for the $\textbf{AF}_e$ of alcohol consumption on postmenopausal invasive breast cancer where we will use case description information on whether the breast cancer is hormone sensitive or hormone insensitive. 

Alcohol consumption is thought to be a risk factor for breast cancer, and various mechanisms have been hypothesized. While some mechanisms would apply to both hormone sensitive and hormone insensitive invasive breast cancer, the following potentially important mechanism would apply only to hormone sensitive mechanisms. Alcohol consumption increases the risk of various hormones such as estrogen. Hormone sensitive breast cancer cells have receptors that allow them to use hormones to grow, thus increased levels of hormones from alcohol consumption may increase the chance of hormone sensitive invasive breast cancer \citep{liu2015links}. \citet{li2010alcohol} provide evidence that alcohol consumption is more strongly associated with hormone sensitive breast cancer than hormone insensitive breast cancer using data from the Women's Health Initiative Observational Study. 

Prior work that has examined the $\textbf{AF}_e$ of alcohol consumption on breast cancer such as \citet{tseng1999calculation} has not leveraged any distinction between hormone sensitive and hormone insensitve breast cancer. The goal of this paper is to develop a method for leveraging this case description information to reduce sensitivity to bias from unmeasured confounding in inference for the $\textbf{AF}_e$.  We will examine data from the Women's Health Initiative (WHI) prospective cohort Observational Study (OS) in 2014 in Section \ref{sec: case-studies}.

Using case description information may introduce selection bias because the treatment assignment might affect the type of case a unit is.  In the motivating example, there would be selection bias if drinking 18 or more oz alcohol a week could cause a woman who would have gotten hormone insensitive breast cancer if she hadn't drunk 18 or more oz drinks a week to instead get hormone sensitive breast cancer.  To account for the possibility of selection bias, we add a sensitivity parameter to the sensitivity analysis.  

\subsection{Prior Work and Outline}

In a case-referent study which compares cases who have some disease to referents who do not have the disease \citep{mantel1973synthetic, breslow1980statistical, holland1987causal}, \citet{small2013case} proposed a framework further defining a case to be a broad case or a narrow case by the magnitude of the response of each case. The proposed methodology can reduce sensitivity to bias by matching narrow cases to controls. This is because, under certain simple models, a narrower case definition yields a greater design sensitivity with fewer cases and is hence less sensitive to the hidden bias in a large sample. Therefore, testing procedures such as the Mantel-Haenszel and Aberrant rank test are less sensitive to hidden bias when applied to narrow cases than broad cases. A similar phenomenon exists when there is effect modification which is when there is a different treatment in different subgroups defined by a pretreatment covariate(s).  In the presence of effect modification, it is possible that the results are less sensitive to bias in a subgroup experiencing larger treatment effects. \citet{hsu2013effect} discovered that when the magnitude of effect changes with an observed covariate that is controlled by matching with moderate stability, the results may be less sensitive to bias for a subset of pairs defined by the observed covariate. This is saying that assuming other things are equal, a greater effect tends to be less sensitive to hidden bias than a smaller effect. However, these articles only consider selecting units based on characteristics that are unaffected by the exposure.  More and more attention has been focused on the danger of selection bias when units are selected based on characteristics that are affected by the exposure. Such research can be found in \citet{zhao2020note} and \citet{knox2020administrative} for example.

The rest of this article is organized as follows. Section \ref{sec: review} introduces the setting and notation as well as a review of the attributable effect and a sensitivity analysis framework for unmeasured confounding. Section \ref{sec: method} proposes our methodology to conduct inference on the attributable fraction among the exposed and a sensitivity analysis framework for selection bias. Section \ref{sec: simulation} demonstrates the simulation results for a power study. We implement our proposed method with the motivating example in Section \ref{sec: case-studies}. Section \ref{sec: conclusion} presents a discussion. Proofs and additional simulation results can be found in the Appendix. Code is available on \url{https://github.com/dosen4552/Attributable-Fraction-Sensitivity-to-Bias}.

\section{Notation and Review: Attributable Effect and Sensitivity Analysis for Hidden Bias} 
\label{sec: review}

\subsection{Notation and Setting}

We consider the setting of matched case-referent studies. In the population of interest, there are $I$ matched sets where the case in each matched set $i = 1, 2, \cdots, I$ is matched with fixed $J - 1 \geq 1$ controls. We use $ij, i = 1, \cdots, I, j = 1, \cdots, J$ to index the $j$-th unit in the $i$-th matched set, and $I \times J = N$. Let $\mathbf{x}_{ij}$ and $ Z_{ij}$ denote the observed covariates and the treatment indicator for unit $ij$, respectively.

Unit $ij$ has two potential outcomes, $r_{T_{ij}}$ and $r_{C_{ij}}$. They are the outcomes that unit $ij$ would have under the treatment and control, respectively. Therefore, $r_{T_{ij}}$ is observed if $ Z_{ij} = 1$ and $r_{C_{ij}}$ is observed if $ Z_{ij} = 0$ (\citet{fisher1923studies}; Neyman (1923)). We further denote the observed outcome $R_{ij} = Z_{ij} r_{T_{ij}} + (1 - Z_{ij}) r_{C_{ij}}$, and let $Z_{i^+} = \sum_{j=1}^J Z_{ij}$ be the number of treated units in matched set $i$. In matched case-referent studies, $\sum_{j = 1}^J R_{ij} = 1$ for every matched set $i$. Without loss of generality, we assume the first unit in each matched set is a case (i.e., $R_{i1} = 1$ for all $i$). The vectors $\bm Z = \{Z_{11},\cdots,Z_{IJ} \}$ and $\bm R = (R_{11}, \cdots, R_{IJ})$ are both observed and they represent the treatment assignment and the observed outcome of all units. $\mathcal{Z}$ is the set of all possible values of $\bm Z$ with $\sum_{j=1}^J Z_{ij} = Z_{i^+}$ for $i = 1, \cdots, I$ and $|\mathcal{Z}| = \prod_{i=1}^I \binom{J}{Z_{i^+}}$. Also, we denote $u_{ij}$ as an unobserved normalized covariate for unit $ij$ (the normalization is that $0 \leq u_{ij} \leq 1$), $\bm u_i = (u_{i1}, u_{i2}, \cdots, u_{iJ})$  and $\mathcal{F} = \{\mathbf{x}_{ij}, u_{ij}, i = 1, \cdots, I, j = 1, \cdots, J \}$ as the set of \emph{fixed} observed and unobserved covariates. We assume matching has controlled the observed covariates in the sense that $\mathbf{x}_{ij} = \mathbf{x}_{ij'}$ for $1 \leq j < j' \leq J$ in each matched set $i$, but matched subjects could differ in their unobserved normalized covariates $u_{ij}$. 

Throughout this paper, we follow the same assumption as \citet{rosenbaum2001effects}; that is, the treatment has nonnegative effect (i.e. $r_{T_{ij}} \geq r_{C_{ij}}$). In the breast cancer study, we consider an outcome of 1 is invasive breast cancer and the treatment (exposure) is drinking alcohol 18+ oz per week (with the control being not drinking alcohol 18+ oz per week), so the nonnegative treatment effect assumption is saying that drinking alcohol 18+ oz per week compared to not drinking alcohol 18+ per week never prevents breast cancer. 

We consider sign-score statistics which are defined below: 

\begin{definition}
A sign-score statistic for a binary response has the following form
\begin{align*}
        t(\mathbf{Z}, \mathbf{R}) = \sum_{i=1}^I \sum_{j = 1}^{J} R_{ij} Z_{ij} = \sum_{i = 1}^I  C_i, \qquad \text{where $C_i = \sum_{j=1}^{J}R_{ij}Z_{ij}$}.
\end{align*}
\end{definition}

\subsection{Sensitivity Analysis: Hidden Bias} \label{subsec: sens}

In observational studies, if a study is free of hidden bias, then two units with the same observed covariates have the same probability of receiving treatment. There is \emph{hidden bias} if two units with the same covariates have different probability of receiving treatment (i.e. $\mathbf{x}_{ij} = \mathbf{x}_{ij'}$ but $\mathbb{P} (Z_{ij} = 1 | \mathcal{F}, \mathcal{Z}) \neq \mathbb{P}(Z_{ij'} = 1 | \mathcal{F},\mathcal{Z})$ for $i = 1, \cdots, I, 1 \leq j < j' \leq J$). A sensitivity analysis asks how the causal conclusion would be altered by hidden bias of various magnitudes. We adapt the sensitivity analysis framework from \citet{rosenbaum2002overt}. Let $\pi_{ij} = \mathbb{P} (Z_{ij} = 1 | \mathcal{F})$. The odds that unit $ij$ and unit $ij'$ receive the treatment are $\pi_{ij} / (1 - \pi_{ij})$ and $\pi_{ij'} / (1 - \pi_{ij'})$, respectively. We consider a sensitivity parameter $\Gamma \geq 1$ that bounds the amount of hidden bias:
\begin{equation} \label{odds: hidden bias}
    \frac{1}{\Gamma} \leq \frac{\pi_{ij} (1 - \pi_{ij'})}{\pi_{ij'} (1 - \pi_{ij})} \leq \Gamma \qquad \text{for $i = 1, \cdots, I, 1 \leq j < j' \leq J$ with $\mathbf{x}_{ij} = \mathbf{x}_{ij'}$}. 
\end{equation}
Here,  $\Gamma$ is a measure of how much unobserved confounding is present. The sensitivity value (the smallest $\Gamma$ at which the p-value would become 0.05) is a measure of how much unobserved confounding would need to be present to alter the conclusions from the observed data \citep{zhao2018sensitivity}. \citet{rosenbaum2002overt} also proposed the sensitivity model
\begin{equation} \label{model: hidden bias}
    \log (\frac{\pi_{ij}}{1 - \pi_{ij}}) = \alpha (\mathbf{x}_{ij}) + \log (\Gamma) u_{ij} \qquad \text{with $0 \leq u_i \leq 1$}
\end{equation}
where $\alpha (\cdot)$ is an unknown function. Proposition 12 in Chapter 4 from \citet{rosenbaum2002overt} states that (\ref{odds: hidden bias}) is equivalent to the existence of model (\ref{model: hidden bias}). Note that when $\Gamma = 1$, unit $ij$ and unit $ij'$ with the same observed covariates have the same odds of receiving treatment, so there is no hidden bias.

In matched case-referent studies, under Fisher's sharp null $H_0: r_{T_{ij}} = r_{C_{ij}}$ and sensitivity model (\ref{model: hidden bias}),  \citet{rosenbaum1991sensitivity} shows that
\begin{align*}
     \tilde{p}_i = \frac{ Z_{i^+}  }{Z_{i^+} + (J - Z_{i^+})\Gamma} \leq \mathbb{P}(C_i = 1 | \mathcal{F}, \mathcal{Z} ) \leq \frac{\Gamma Z_{i^+} }{\Gamma Z_{i^+} + (J - Z_{i^+})} = \tilde{\tilde{p}}_i.
\end{align*}
The bounds are sharp in the sense that the upper bound is attained at $\bm u_i = (1,0, \cdots, 0)$ and the lower bound is attained at $\bm u_i = (0, 1, \cdots, 1)$. Let $\tilde{C}_i$ be the Bernoulli random variable with success probability $\tilde{p}_i$, and $\tilde{\tilde{C}}_i$ be the Bernoulli random variable with success probability $\tilde{\tilde{p}}_i$, then under $H_0: r_{T_{ij}} = r_{C_{ij}}$, and let $T = t(\mathbf{Z}, \mathbf{R})$, for every $k$, 

\begin{equation} \label{equ: bound on p-value}
    \mathbb{P}(\sum_{i=1}^I \tilde{C}_i  \geq k )  \leq \mathbb{P}(T \geq k | \mathcal{F}, \mathcal{Z} ) \leq \mathbb{P}(\sum_{i=1}^I \tilde{\tilde{C}}_i  \geq k ) .
\end{equation}

Equation (\ref{equ: bound on p-value}) yields sharp bounds on the $P$-value testing $H_0: r_{T_{ij}} = r_{C_{ij}}$ in the presence of unmeasured confounding bias controlled by $\Gamma$. 

\subsection{Review of the Attributable Effect}

The attributable effect is a way to measure the magnitude of effect of a treatment on a binary outcome. It is denoted by $A$ and defined by $A = \sum_{i=1}^I \sum_{j=1}^J Z_{ij} (r_{T_{ij}} - r_{C_{ij}})$, that is the number of treated events actually caused by exposure to the treatment \citep{rosenbaum2001effects, rosenbaum2010design}. In our breast cancer study, the attributable effect is the number of invasive breast cancer cases among women drinking alcohol 18+ oz per week that would not have occurred if those women did not drink alcohol 18+ oz per week. The attributable fraction among the exposed (\textbf{AF}$_e$) is the attributable effect divided by the number of treated cases.

\begin{table}
\caption{\label{table:1} A generic example of $2 \times 2$ contingency table for a pair-matched case-referent study. $+$ and $-$ indicate if a unit in the matched pair was treated or not. }
\centering
\begin{tabular}{ c c c }
\hline
\hline
 Outcome & Control$+$ & Control$-$ \\
 \hline 
 Case$+$ & $a$ & $b$ \\ 
 Case$-$ & $ c$ & $ d$ \\  
 \hline
\end{tabular}

\end{table}

Let $\bm \delta = (\delta_{11}, \cdots, \delta_{IJ})$ be the vector of treatment effects where $\delta_{ij} = r_{T_{ij}} - r_{C_{ij}}$, then $A = \sum_{i=1}^I \sum_{j = 1}^J Z_{ij} \delta_{ij}$. Consider a particular null hypothesis, $H_0: \bm \delta = \bm \delta_0$ for some specific vector $\bm \delta_0 = (\delta_{0_{11}}, \cdots, \delta_{0_{IJ}})$. We call $\bm \delta_0$ ``compatible'' if $\delta_{0_{ij}} = 0$ whenever $(R_{ij}, Z_{ij}) = (0, 1)$ or $(R_{ij}, Z_{ij}) = (1, 0)$ for all $i, j$ because of the nonnegative assumption, and ``incompatible'' otherwise \citep{rosenbaum2001effects}. To understand ``compatible'', since $R_{ij} = Z_{ij} r_{T_{ij}} + (1 - Z_{ij}) r_{C_{ij}}$, if $Z_{ij} = 1$ and $R_{ij} = 0$, then $r_{T_{ij}} = r_{C_{ij}} = 0$ because $r_{T_{ij}} \geq r_{C_{ij}}$. Hence, $\delta_{0_{ij}} = 0$. Similar reason holds for $Z_{ij} = 0$ and $R_{ij} = 1$ because $r_{T_{ij}} = r_{C_{ij}} = 1$.  If the sharp null hypothesis $\bm \delta = \bm \delta_0$ is incompatible, reject it (the type I error rate for incompatible hypotheses is then zero).  If the null hypothesis  $H_0: \bm \delta = \bm \delta_0$ is true, then $r_{C_{ij}}$ will be known for all $i, j$ since $r_{C_{ij}} = R_{ij} - Z_{ij} \delta_{0_{ij}}$. Let $A_0 = \sum_{i=1}^I\sum_{j=1}^J Z_{ij} \delta_{0_{ij}}$, then under the null hypothesis, $T - A_0 = \sum_{i=1}^I \sum_{j=1}^J Z_{ij}r_{C_{ij}}$. If we further denote $B_i = \sum_{j=1}^J Z_{ij} r_{C_{ij}}$ and $r_{C_{i^+}} = \sum_{j=1}^J r_{C_{ij}}$ for the number who would have had events if exposure to treatment had been prevented, then

\begin{align*}
  \tilde{\pi}_i = \frac{ Z_{i^+} r_{C_{i^+}}  }{ Z_{i^+} r_{C_{i^+}} +  \Gamma (J - Z_{i^{+}} r_{C_{i^+}} )  }  \leq \mathbb{P} (B_i = 1) \leq \frac{\Gamma Z_{i^+} r_{C_{i^+}}  }{ \Gamma Z_{i^+} r_{C_{i^+}} + J - Z_{i^{+}} r_{C_{i^+}}  }  = \tilde{\tilde{\pi}}_i,
\end{align*}

and, in particular, in a randomized experiment when $\Gamma = 1$, $\mathbb{P} (B_i = 1) =Z_{i^+} r_{C_{i^+}}/J$. Note that the attributable effect is a random variable rather than a parameter, however, we can still conduct the inference on it by constructing a one-sided prediction interval for $A$.  

To construct a one-sided prediction interval for the attributable effect $A$, starting with $A_0 = 0$, approximate the upper bound of $\mathbb{P} \left(\sum_{i=1}^I B_i \geq T -A_0 \right) $ by

\begin{equation}
\label{equ: approx upper}
                    1 - \Phi \left( \frac{T -A_0 - \sum_{i=1}^I \tilde{\tilde{\pi}}_i }{\sqrt{\sum_{i=1}^I \tilde{\tilde{\pi}}_i (1 - \tilde{\tilde{\pi}}_i)  } }     \right)
\end{equation}
when $I \rightarrow \infty$, and $\Phi (\cdot)$ is the cumulative density function of the standard normal distribution. If (\ref{equ: approx upper}) is less than $\alpha$, $A_0$ should be increased by one until it is not.  Assume that $A_0 = a_*$,  then  $\{A: A \geq a_*\}$ is a one-sided $100 \times (1 - \alpha) \% $ minimum prediction interval for the attributable effect $A$. To better understand, we use a generic example for pair-matched case-referent study from Table \ref{table:1} to demonstrate the calculation. There are $ a + b$ treated cases, and $b$ of these are in discordant pairs. Suppose that the hypothesis $H_0: \bm \delta = \bm \delta_0$ results in exactly $a^*$ cases that are attributable to the exposure, and all $a^*$ are among $b$ discordant pairs, then (\ref{equ: approx upper}) can be re-written as:
\begin{align*}
            1 - \Phi \left( \frac{ ( b - a_*) -  (b+c - a_*) p_{\Gamma}  }{\sqrt{(b+c - a^*)  (1 - p_\Gamma) p_{\Gamma}  } }     \right)
\end{align*}
where $p_{\Gamma} = \frac{\Gamma}{1 + \Gamma}$. Details of the procedure of finding the attributable effect for multiple-controls matched case-referent studies can be found in the Appendix III. 

For all $\bm \delta_0$ that correspond to a given attributable effect $a$, we can find a $\bm \delta_0^*$ that is the asymptotic worst case and provides an asymptotically correct maximum $P$-value over all possible $\bm \delta_0$ that correspond to the given attributable effect $a$ \citep{gastwirth2000asymptotic}. Our setting is a special case ($n_i = 2, R_{i^+} = 1$ for every $i$) of the general procedure demonstrated in \citet{rosenbaum2002attributing} Section 3.4. Detail discussion on the asymptotic separability can be found in the Appendix II.


This observation may lead one to describe the test as evaluating whether the attributable effect, $\sum_{i,j} Z_{ij} \delta_{ij}$, equals its hypothesized value, $\sum_{i,j} Z_{ij} \delta_{0_{ij}}$. However, it is crucial to emphasize again that the attributable effect, $\sum_{i,j} Z_{ij} \delta_{ij}$, is a random variable rather than a fixed parameter, precluding direct hypothesis testing about its value. Nevertheless, statistical inference on the attributable effect can still be conducted by constructing a one-sided $(1 - \alpha) \times 100\%$ prediction interval. This interval has the interpretation that, in repeated studies, the attributable effect will fall within the interval $(1-\alpha) \times 100\%$ of the time. 


This prediction interval for the attributable effect can be thought of as providing a summary for the $N$-dimensional confidence set for $\bm \delta_0$ which would be difficult to interpret.  A detailed justification can be found in Appendix I.

\section{Robust Inference for the Attributable Fraction Among the Exposed}


\label{sec: method}

Without loss of generality, let $\mathcal{C}_k, 1 \leq k \leq L$ be the set of indices of matched sets with $k$-th subtype of cases. In the context of our motivating example, $L = 2$, and $\mathcal{C}_1$ represents the set of indices of matched sets with hormone sensitive invasive breast cancer cases while $\mathcal{C}_2$ represents the set of indices of matched sets with hormone insensitive invasive breast cancer cases. The fundamental idea of our proposed method is to conduct the hypothesis testing on $\mathcal{C}_1$ and $\mathcal{C}_2$ separately (denoted as $H_1, H_2$), and then combine the $P$-values to test the plausibility of both nulls holding (denoted as $H_\wedge := H_1 \wedge H_2 $ ).

\subsection{Bonferroni, Fisher's Method, and Stouffer's Z-Score Method}

Let $P_1, \cdots, P_L$ be valid, statistically independent $P$-values testing hypothesis $H_1,  \cdots, H_L$, respectively. Hence, $\mathbb{P}(P_k \leq \alpha) \leq \alpha$ for all $\alpha \in [0,1]$ if $H_k$ is true for $k = 1, 2, \cdots, L$. In the context of this paper, $N$ units have been partitioned into $L = 2$ non-overlapping groups based on the status of breast cancer cases (i.e. hormone sensitive breast cancer and hormone insensitive breast cancer).  The conjunction $H_\wedge = H_1 \wedge \cdots \wedge H_L$ asserts that all the hypotheses are true. The Bonferroni method can be expressed using the formula $L \min \{P_1, \cdots, P_L\}$. However, the Bonferroni method tends to be conservative because it does not make any assumptions about the structure of the p-values.  For independent p-values, \citet{fisher1970statistical} proposed the method of combining $L$ independent $P$-values, into one test statistic by $-2 \log (\prod_{k=1}^L P_k) = -2 \sum_{k=1}^L \log (P_k)$ which is stochastically smaller than the chi-square distribution with degree of freedom $2L$, and it is an exact chi-square distribution with degree of freedom $2L$ if $\mathbb{P}(P_k \leq \alpha) = \alpha$ when $H_k$ is true.  
Although in most cases Fisher's method is less conservative than the Bonferroni method, it may sometimes be quite conservative when $H_\wedge$ is false but many $H_k$'s are true with $\mathbb{P}(P_k \leq \alpha) < \alpha$ since $\prod_{k=1}^L P_k$ is excessively large. With this issue in mind, \citet{zaykin2002truncated} proposed a truncated product of $P$-values, $P_\wedge = \prod_{k=1}^L P_k^{\mathbf{1}\{P_k \leq \tilde{\alpha} \} }$ to test $H_\wedge$ where $\mathbf{1}\{P_k \leq \tilde{\alpha}\}$ equals $1$ when $P_k \leq \tilde{\alpha}$ and $0$ otherwise. Note that when $\tilde{\alpha} = 1$, this is exactly Fisher's method. \citet{zaykin2002truncated} give the distribution of $P_\wedge$ when $\mathbb{P}(P_k \leq \alpha) = \alpha$ which yields the reference distribution for $\mathbb{P}(P_k \leq \alpha) \leq \alpha$. If the $P_k$s are independent, under the null hypothesis,  
\begin{equation} \label{equ: trunc}
\begin{split}
     \mathbb{P} (P_\wedge <  w) &= \sum_{k=1}^L \mathbb{P}(P_\wedge \leq  w | K = k) \mathbb{P}(K = k) \\
     & = \sum_{k=1}^L \binom{L}{k} (1 - \tilde{\alpha})^{L-k} \left( w \sum_{s = 0}^{k-1} \frac{(k \log(\tilde{\alpha}) - \log(\tilde{\alpha}))^s }{s!} \mathbf{1} \{ w \leq \tilde{\alpha}^k\} + \tilde{\alpha}^k \mathbf{1}\{w > \tilde{\alpha}^k \} \right).
\end{split}
\end{equation}

Moreover, if the $P_k$s are independent and stochastically larger than the uniform distribution on the $L$-dimensional unit cube $[0,1]^L$, then $\mathbb{P}(P_\wedge < w)$ would be less than or equal to the right side of (\ref{equ: trunc}). Additional methods for combining dependent $P$-values can be found in \citet{rosenbaum2011new} and \citet{brannath2002recursive}. Another analogous approach to Fisher's method is Stouffer's Z-score method. If we let $Z_k = \Phi^{-1} (1 - P_k) $, then

\begin{align*} 
    Z \sim \frac{\sum_{k=1}^L Z_k}{ \sqrt{L} }
\end{align*}
is a Z-score \citep{zaykin2011optimally}, and it can be used to obtain the one-sided combined $P$-value. There are several advantages of using the Z-score method.  First, it is relatively straightforward to introduce weights based on different sizes of $C_k$s:

\begin{align*}
     Z \sim \frac{\sum_{k=1}^L w_k Z_k}{ \sqrt{\sum_{k=1}^L w_k^2} }
\end{align*}
where $w_k = \sqrt{|\mathcal{C}_k|} $ \citep{liptak1958combination}. Second, the power of sensitivity analysis can be easily assessed. 


\subsection{Sensitivity Model for Selection Bias}

The implementation of the above methods combining $P$-values can be more powerful in the sense of the robustness for hidden bias. However, in general, restricting to a particular subtype of cases may introduce additional selection bias under the $H_0: r_{T_{ij}} = r_{C_{ij}}$ because even if the exposure does not affect whether a person has the disease, it may affect the type of disease. To develop a valid testing procedure for $H_0: r_{T_{ij}} = r_{C_{ij}}$ for $\{ij: i \in \mathcal{C}_k \}$ in matched case-referent studies, we adapt the following sensitivity model by \citet{ye2021combining} for selection bias (although we acknowledge that there are alternative approaches to modeling selection bias such as described in  \citet{little2019statistical}). Let $\kappa_k (\cdot) = 1$ denote the $k$-th subtype of cases and $\kappa_k (\cdot) = 0$ denote the non-$k$-th subtype of cases. Individuals have two potential outcomes for the $k$-th subtype of cases: $(\kappa_k (r_{T_{ij}}),\kappa_k (r_{C_{ij}}))$, and $\kappa_k (R_{ij}) = Z_{ij}\kappa_k (r_{T_{ij}}) + (1 - Z_{ij})\kappa_k (r_{C_{ij}}) $ is the observed outcome for $k$-th subtype of cases. Then the sensitivity model is:

\begin{equation} \label{model: selection bias}
    \begin{split}
        \pi_{ij} &:= \mathbb{P}(Z_{ij} = 1|\mathbf{x}_{ij}, u_{ij}) = \frac{\exp \left( \alpha (\mathbf{x}_{ij}) + \log (\Gamma) u_{ij} \right) }{\sum_{j=1}^J \exp \left( \alpha (\mathbf{x}_{ij}) + \log (\Gamma) u_{ij} \right) } \\
        \theta_{T_{ij}^{(k)} } &:= \mathbb{P} \left(  \kappa_k (r_{T_{ij}}) = 1 | \mathbf{x}_{ij}, u_{ij}, r_{T_{ij}} = r_{C_{ij}} = 1 \right) \\
        \theta_{C_{ij}^{(k)} } &:= \mathbb{P} \left( \kappa_k (r_{C_{ij}}) = 1 | \mathbf{x}_{ij}, u_{ij},  r_{T_{ij}} = r_{C_{ij}} = 1 \right) \\
    \end{split}
\end{equation}

where $0 \leq u_{ij} \leq 1, 1 \leq \theta_{T_{ij}^{(k)} }/\theta_{C_{ij}^{(k)} } \leq \Theta, \theta_{C_{ij}^{(k)}} > 0, \Theta, \Gamma \geq 1, \text{ for every } k = 2, \cdots, L,$ and $\sum_{k=1}^L \theta_{T_{ij}^{(k)} } = 1$.
In model (\ref{model: selection bias}), we have the sensitivity parameter $\Gamma$ from Rosenbaum's sensitivity model (\ref{model: hidden bias}). We also have the additional sensitivity parameter $\Theta$ to bound $\theta_{T_{ij}^{(k)} }/\theta_{C_{ij}^{(k)} }$. The natural interpretation of $\Theta$ is that it measures the maximum relative risk of every subtype of cases under the exposure compared to control. Under model (\ref{model: selection bias}), the following propositions for testing $H_0: r_{T_{ij}} = r_{C_{ij}}$ for $\{ij: i \in \mathcal{C}_k\}$ can be stated. Proofs are in the Appendix.

\begin{figure}
  \caption{\label{fig:diag} A directed acyclic graph (DAG) under $H_0: r_{T_{ij}} = r_{C_{ij}}$ for $\{ij: i \in \mathcal{C}_1\}$ where $\mathcal{C}_1$ is the set of indices of matched sets that contain hormone-sensitive invasive breast cancer cases in our motivating example. $X$ is the vector of observed covariates (e.g. demographics) while $U$ is the unmeasured confounding (e.g. exposure to radiation ever). Restricting to hormone sensitive invasive breast cancer may create selection bias under $H_0$ because drinking alcohol 18+ oz per week may change hormone sensitive invasive breast cancer to hormone non-sensitive invasive breast cancer.}
  \centering
  \includegraphics[width=10cm]{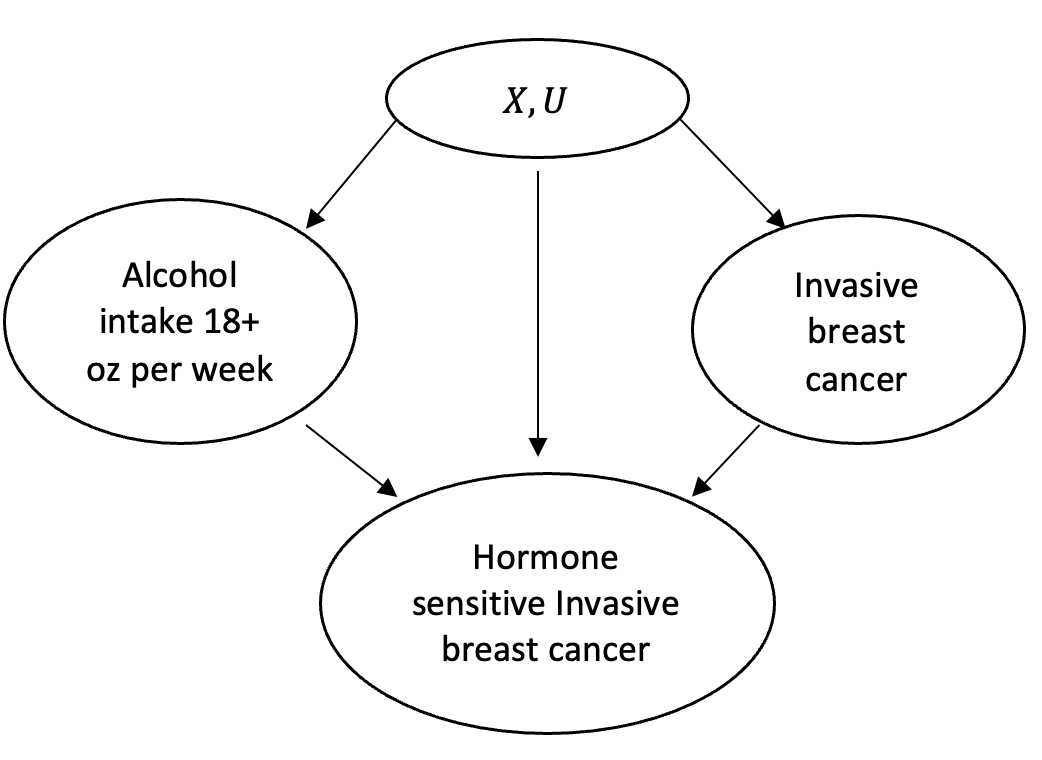}
\end{figure}

\begin{table}
\caption{\label{table:4} The observed $2 \times 2$ contingency table with a sign-score statistic for $i$-th matched set.}
\centering
\begin{tabular}{|| c c c c||}
\hline
  & $Z_{ij} = 1$ & $Z_{ij} = 0$ & Total\\
 \hline\hline
 $R_{ij} = 1$ & $C_i$ & $1 - C_i$ & $1$\\ 
 $R_{ij} = 0$ & $Z_{i^+} - C_i$ & $J-Z_{i^+}-1+C_i$ & $J-1$\\  
 Total & $Z_{i^+}$ & $J-Z_{i^+}$  & $J$   \\
 \hline
\end{tabular}
\end{table}

\begin{proposition} [\cite{ye2021combining}] \label{prop: ith matched set}
Under the null hypothesis that there is no treatment effect and for each fixed $\Gamma$ and $\Theta$
\begin{equation} \label{equ: ith matched set}
    \bar p_{i}=\frac{Z_{i^+}}{Z_{i^+} + (J-Z_{i^+})  \Gamma } \leq  \mathbb{P} (C_{i}=1 | \mathcal{F}, i \in \mathcal{C}_k, \mathcal{Z}) \leq \frac{ \Gamma \Theta Z_{i^+} }{\Gamma \Theta Z_{i^+} + (J - R_{i^+}) } = \bar{\bar p}_{i}.
\end{equation}
The bounds are sharp in the sense that the upper bound is attained at $\mathbf{u}_{i} = (1,0,\cdots,0)$ and $\theta_{T_{ij}^{(k)} }/\theta_{C_{ij}^{(k)} } = \Theta$, and the lower bound is attained at $\mathbf{u}_{i} = (0,1,\cdots,1)$ and $\theta_{T_{ij}^{(k)} }/\theta_{C_{ij}^{(k)} } = 1$. 
\end{proposition}

 Proposition \ref{prop: ith matched set} states the bounds on the $P$-value for the test statistic of the upper left corner cell in Table \ref{table:4}. Now consider the test statistic $\sum_{i \in \mathcal{C}_k} C_i $  which is the sum of $| \mathcal{C}_k |$ independent Bernoulli random variables taking the value $1$ with probability bounded by $ \bar p_{i}$ and $\bar{\bar p}_{i}$. The next proposition gives the bound of distribution of $\sum_{i \in \mathcal{C}_k} C_i $, which is an immediate result from Proposition \ref{prop: ith matched set} and some reviewed materials from Section \ref{subsec: sens}.

\begin{proposition} \label{prop:score}
Under the null hypothesis that there is no treatment effect, then for each $\Gamma, \Theta \geq 1$ and for all $k$, 
\begin{equation}
            \mathbb{P}(\sum_{i \in \mathcal{C}_k} \bar{C}_i \geq q) \leq \mathbb{P} (\sum_{i \in \mathcal{C}_k} C_i \geq q | \mathcal{F}, i \in \mathcal{C}_k, \mathcal{Z}) \leq \mathbb{P}(\sum_{i\in \mathcal{C}_k} \bar{\bar C}_i \geq q) 
\end{equation}
where $\bar{C}_i, \bar{\bar C}_i$ are independent Bernoulli random variables with 
probability $\bar{p}_i$ and $\bar{\bar p}_i$ equal to $1$, respectively.  
\end{proposition}

 With the above propositions, we extend the results to attributable effects. Let $ \bar{\bm \delta}_i = (\bar{\delta}_{i1}, \bar{\delta}_{i2}, \cdots, \bar{\delta}_{iJ} )$, $\bar{\bm \delta}_{0_i} = (\bar{\delta}_{0_{i1}}, \bar{\delta}_{0_{i2}}, \cdots, \bar{\delta}_{0_{iJ}} )$ and  $B_i = \sum_{j=1}^J Z_{ij} r_{C_{ij}}$, then we have the following proposition for attributable effects. 

\begin{proposition} \label{prop:attri}
Write $r_{C_{i^+}} = \sum_{j=1}^J r_{C_{ij}}$. Under the null hypothesis $H_0: \bar{\bm \delta}_i = \bar{\bm \delta}_{0_i}$ and for each fixed $\Gamma$, $\Theta \geq 1$
\begin{equation}
    \bar \pi_{i}=\frac{Z_{i^+} r_{C_{i^+}} }{Z_{i^+}  r_{C_{i^+}} + (J-Z_{i^+}  r_{C_{i^+}})  \Gamma } \leq  \mathbb{P} (B_{i}=1 | \mathcal{F}, i \in \mathcal{C}_k, \mathcal{Z}) \leq \frac{\Theta \Gamma  r_{C_{i^+}}}{\Gamma \Theta Z_{i^+}   r_{C_{i^+}} + (J-Z_{i^+} r_{C_{i^+}})} = \bar{\bar \pi}_{i}.
\end{equation}
The bounds are sharp in the sense that the upper bound is attained at $\mathbf{u}_{i} = (1,0,\cdots,0)$ and $\theta_{T_{ij}^{(k)} }/\theta_{C_{ij}^{(k)} } = \Theta$, and the lower bound is attained at $\mathbf{u}_{i} = (0,1,\cdots,1)$ and $\theta_{T_{ij}^{(k)} }/\theta_{C_{ij}^{(k)} } = 1$. 
\end{proposition}

Proposition \ref{prop:attri} is an immediate consequence of Proposition \ref{prop: ith matched set} and equation (2) from \citet{rosenbaum2002attributing}. Analogously, we have the following proposition for the attributable effects within $\mathcal{C}_k$, which is the number of cases from the $k$-th subtype that is actually attributed to the exposure or treatment. 

\begin{proposition} 
Under the joint null hypothesis $H_0: \bar{\bm \delta}_i = \bar{\bm \delta}_{0_i}$ for $i \in \mathcal{C}_k$ and for each fixed $\Gamma$, $\Theta \geq 1$
\begin{equation} \label{equ: attri}
    \mathbb{P}(\sum_{i \in \mathcal{C}_k} \bar{D}_i \geq k) \leq  \mathbb{P} (\sum_{i \in \mathcal{C}_k} B_{i} \geq k | \mathcal{F}, i \in \mathcal{C}_k, \mathcal{Z}) \leq \mathbb{P}(\sum_{i \in \mathcal{C}_k} \bar{\bar D}_i \geq k). 
\end{equation}
where $\bar{D}_i, \bar{\bar D}_i$ are independent Bernoulli random variables with 
probability $\bar{\pi}_i$ and $\bar{\bar \pi}_i$ equal to $1$, respectively.  
\end{proposition}

\subsection{Design Sensitivity of the Methods for Combination of $P$-Values}
The power of one-sided $\alpha$ sensitivity analysis is the probability that the upper bound of equation (\ref{equ: ith matched set}) - (\ref{equ: attri}) is less than or equal to $\alpha$, i.e., the probability that we find evidence for a casual effect that is insensitive to bias $(\Gamma ,\Theta)$. We consider power under the favorable situation in which there is a true causal effect and there is in fact no bias. In general, the favorable situation refers to the situation in which $H_0: \bm \delta = \bm 0$ is false and $\mathbb{P}(\mathbf{Z} = \mathbf{z}| \mathcal{F}, ij \in \mathcal{C}_1, \mathcal{Z}) = 1/J^I$. In the favorable situation, we unambiguously would like to reject the null hypothesis of no treatment effect \citep{hansen2014clustered}. For a fixed $\Theta^*$, the design sensitivity is the sensitivity parameter $\tilde{\Gamma}$ such that the power of sensitivity analysis tends to $1$ for any $\Gamma < \tilde{\Gamma}$ and $0$ for any $\Gamma > \tilde{\Gamma}$ when $I \rightarrow \infty$. So in the limit, we are able to distinguish a treatment effect without bias when $\Gamma < \tilde{\Gamma}$ but not when $\Gamma > \tilde{\Gamma}$. The following proposition states the design sensitivity parameter for methods combining $P$-values. 

\begin{proposition} \label{prop:design sensitivity}
Suppose there are $L$ non-overlapping subgroups of matched sets defined by different types of positive outcomes (i.e. different subtypes of cases), and allow the number of matched sets $I \rightarrow \infty$ with $L$ fixed, in such a way that the fraction of matched sets in each subgroup is tending to a nonzero constant. Suppose for a fixed $\Theta^*$, the test in subgroup $l$ has design sensitivity $\tilde{\Gamma}_l$. Then for $0 < \tilde{\alpha} \leq 1$, the design sensitivity of various methods for combining $P$-values (e.g. Bonferroni, Fisher, etc) is $\tilde{\Gamma}_{\max} = \max\left(\tilde{\Gamma}_1, \cdots, \tilde{\Gamma}_L \right)$. 
\end{proposition}

\subsection{Power of Sensitivity Analysis}

 Assume that the null hypothesis will be rejected at the significance level $\alpha$, and the alternative hypothesis is true and $A_0 = a_*$, for large $I$ and $|\mathcal{C}_k|$ for $k = 1, \cdots, L$. Using the normal approximation, the statistical power can be approximated by
\begin{align*}
    B(a_*) = 1 - \Phi \left(Z_{1 - \alpha} - \frac{ a_*   }{\sqrt{\sum_{i=1}^I \bar{\bar \pi}_i (1 - \bar{\bar \pi}_i)  } }  \right)
\end{align*}
without considering additional case information. Here, $Z_{1-\alpha}$ is the $Z$-score at $1 - \alpha$ level. When case description information is used,  assume that the null hypothesis will be rejected at the significance level $\alpha$, and the alternative hypothesis is true and $A_0^{(k)} = a_*^{(k)}$ in $\mathcal{C}_k$ for $k = 1, 2, \cdots,L$ , analogously, 

\begin{align*}
    B(a_*^{(1)}, a_*^{(2)}, \cdots, a_*^{(L)})  = 1 - \Phi \left( Z_{1 - \alpha} - \sum_{k=1}^L \frac{ a_*^{(k)} }{ \sqrt{\sum_{i \in \mathcal{C}_k} \bar{\bar \pi}_i (1 - \bar{\bar \pi}_i)  }   } \middle/ \sqrt{L}   \right)
\end{align*}
by Stouffer's Z-score method, and 
\begin{align*}
    B(a_*^{(1)}, a_*^{(2)}, \cdots, a_*^{(L)})  = 1 - \Phi \left( Z_{1 - \alpha} - \sum_{k=1}^L  \frac{ w_k a_*^{(k)} }{ \sqrt{\sum_{i \in \mathcal{C}_k} \bar{\bar \pi}_i (1 - \bar{\bar \pi}_i)  }   } \middle/ \sqrt{\sum_{k=1}^L w_k^2}   \right)
\end{align*}
by weighted Stouffer's Z-score method. The analytical power formulas provide guidance about when using case description information can reduce the sensitivity to hidden bias given fixed $\Theta$. For example, assume $L = 2$ (i.e. there are two subtypes of cases), and if the sum of attributable effects obtained from these subtypes is greater than the attributable effect obtained from the overall population without using case description information, then using the case definition information leads to higher power of sensitivity analysis. Detailed power calculation can be found in the Appendix III.

\section{Simulation Studies}
\label{sec: simulation}

\subsection{Goal and Structure}

Our primary goal in the simulation section is to compare the power of sensitivity analysis across different methods of combining $P$-values from our proposed methodology under a wide range of data-generating processes (DGPs). To be more specific, we estimate the probability that a sensitivity analysis performed with the parameters set to be $\Gamma$ and $\Theta$ produces an upper bound on the one-sided $P$-value of at most $0.05$ in the favorable situation (\citet{hansen2014clustered}) with different treatment effect and no unmeasured confounders.

In our simulation setting, since our outcomes are binary, we generate $r_T \sim \text{Bernoulli}(p_1), r_C \sim \text{Bernoulli}(p_2)$ separately with the same size $n$. The observed outcome is $R = Z \cdot r_T + (1 - Z) \cdot r_C$, and the average treatment effect is $\mathbb{E} (r_T - r_C) = p_1 - p_2 = \delta$. Then we randomly assign each unit to treatment or control with equal probability $1/2$. We apply this DGP to generate two datasets (denoted by $\mathcal{D}_1, \mathcal{D}_2$) separately with the same sample size but different average treatment effect. We summarize these two datasets by using two $2 \times 2$ contingency tables.

To understand this setting better, we refer back to our motivating example. In the context of our motivating example, $\mathcal{D}_1$ refers to the group of patients with hormone sensitive invasive breast cancer and non-cases while  $\mathcal{D}_2$ refers to the group of patients with hormone insensitive invasive breast cancer and non-cases. The treatment (exposure) refers to drinking alcohol 18+ oz per week. We want to compute the power of sensitivity analysis across different combinations of methods and scenarios.  

The factors of our simulation can be described as follows: 

\textbf{Factor 1:} methods of combining $P$-values: (i) directly merge two tables and then conduct hypothesis testing; (ii) Stouffer's $Z$-Score method; (iii) Fisher's method; (iv) truncated Fisher's method with $\tilde{\alpha} = 0.05$ and (v) Bonferroni method.

\textbf{Factor 2:} combinations of average treatment effect for two datasets: (i) equal effects: $\delta_1 = 0.2, \delta_2 = 0.2$; (ii) moderately unequal effects: $\delta_1 = 0.4, \delta_2 = 0.2$; (iii) unequal effects: $\delta_1 = 0.6, \delta_2 = 0.2$ and (iv) effects only in one dataset: $\delta_1 = 0.6, \delta_2 = 0.0$. This factor varies the additional information on two groups of units that have different responses to the treatment or exposure. 

\textbf{Factor 3:} sensitivity parameter $\Gamma$ for hidden bias: (i) $\Gamma = 1.0$; (ii) $\Gamma = 1.5$; (iii) $\Gamma = 2.0$; (iv) $\Gamma = 2.5$; (v) $\Gamma = 3.0$; (vi) $\Gamma = 3.5$;  (vii) $\Gamma = 4.0$; (viii) $\Gamma = 4.5$ and (ix) $\Gamma = 5.0$. 

\textbf{Factor 4:} sensitivity parameter $\Theta$ for selection bias: (i) $\Theta = 1.0$; (ii) $\Theta = 1.1$ and (iii) $\Theta = 1.2$. 

Lastly, we fix our sample size $n = 500$ for each dataset, and therefore the whole merged dataset has sample size $1000$. We repeat the simulation $200$ times and report the proportion of times that the null hypothesis $H_0: \bm \delta = \bm 0$ is rejected.

\begin{table}
\hspace*{-2cm}
\caption{\label{table:5} Power simulations include directly merged method, Stouffer's $Z$-score method, Fisher's method, truncated method with $\tilde{\alpha} = 0.05$, and Bonferronni method. Each situation is sampled $200$ times. There are $500$ data points in each dataset with average treatment effects $\delta_1$ and $\delta_2$. $\Theta = 1$.  }
\centering
\scalebox{0.7}{
\begin{tabular}{||c c c c c c c c c c||} 
 \hline
 \multicolumn{10}{||c||}{Equal effects: $\delta_1 = 0.2, \delta_2 = 0.2$}   \\
 \hline \hline
   & $\Gamma = 1.0$  & $\Gamma = 1.5$  & $\Gamma = 2.0$  & $\Gamma = 2.5$  & $\Gamma = 3$ & $\Gamma = 3.5$ & $\Gamma = 4$ & $\Gamma = 4.5$ & $\Gamma = 5$  \\ 
 \hline
 Merged & 1.000 &  1.000 &    1.000 & 0.990 & 0.910 & 0.565 & 0.315 & 0.105 & 0.040  \\
 Stouffer & 1.000 & 1.000 & 1.000 & 0.990 & 0.890 & 0.580 & 0.325 & 0.100 & 0.035 \\ 
 Fisher & 1.000 & 1.000 & 1.000 & 0.985 & 0.885 & 0.570 & 0.305 & 0.090 & 0.035 \\ 
 Truncated $\tilde{\alpha} = 0.05$ & 1.000 & 1.000 & 1.000 & 0.970 & 0.805 & 0.475 & 0.220 & 0.085 & 0.015 \\
 Bonferroni  & 1.000 & 1.000 & 1.000 & 0.970 & 0.785 & 0.460 & 0.215 & 0.080 & 0.010\\
 \hline \hline

 \multicolumn{10}{||c||}{Slightly unequal effects: $\delta_1 = 0.3, \delta_2 = 0.2$} \\
 \hline \hline
   & $\Gamma = 1.0$  & $\Gamma = 1.5$  & $\Gamma = 2.0$  & $\Gamma = 2.5$  & $\Gamma = 3$ & $\Gamma = 3.5$ & $\Gamma = 4$ & $\Gamma = 4.5$ & $\Gamma = 5$  \\ 
 \hline
 Merged & 1.000 &  1.000 &  1.000 & 1.000 & 0.965 & 0.780 & 0.440 & 0.255 & 0.125 \\
 Stouffer & 1.000 &  1.000 &  1.000 & 1.000 & 0.970 & 0.830 & 0.515 & 0.290 & 0.145 \\ 
 Fisher & 1.000 &  1.000 &  1.000 & 1.000 & 0.955 & 0.805 & 0.465 & 0.250 & 0.145\\ 
 Truncated $\tilde{\alpha} = 0.05$ & 1.00 & 1.000 &  1.000 & 1.000 & 0.910 & 0.675 & 0.445 & 0.255 & 0.145 \\ 
 Bonferroni & 1.000 &  1.000 & 1.000 & 0.995 & 0.905 & 0.675 & 0.430 & 0.250 & 0.145\\
 \hline \hline
 \multicolumn{10}{||c||}{Unequal effects: $\delta_1 = 0.6, \delta_2 = 0.2$} \\
 \hline \hline
   & $\Gamma = 1.0$  & $\Gamma = 1.5$  & $\Gamma = 2.0$  & $\Gamma = 2.5$  & $\Gamma = 3$ & $\Gamma = 3.5$ & $\Gamma = 4$ & $\Gamma = 4.5$ & $\Gamma = 5$  \\ 
 \hline
 Merged & 1.000 &  1.000 &  1.000 &  1.000 & 1.000 &  1.000 & 0.985 & 0.930 & 0.780 \\
 Stouffer & 1.000 &  1.000 & 1.000 &  1.000 &  1.000 &  1.000 & 1.000 & 0.970 & 0.890 \\ 
 Fisher & 1.000 &  1.000 &    1.000 &  1.000 &    1.000 &  1.000 & 1.000 & 1.00 & 1.000 \\ 
 Truncated $\tilde{\alpha} = 0.05$ & 1.000 &  1.000 &  1.000 & 1.000 & 1.000 & 1.000 & 1.000 & 1.000 & 1.000 \\ 
 Bonferroni & 1.000 &  1.000 &  1.000 & 1.000 & 1.000 & 1.000 & 1.000 & 1.000 & 1.000 \\
 \hline \hline
 \multicolumn{10}{||c||}{Effect only in one dataset: $\delta_1 = 0.6, \delta_2 = 0.0$} \\
 \hline \hline
   & $\Gamma = 1.0$  & $\Gamma = 1.5$  & $\Gamma = 2.0$  & $\Gamma = 2.5$  & $\Gamma = 3$ & $\Gamma = 3.5$ & $\Gamma = 4$ & $\Gamma = 4.5$ & $\Gamma = 5$  \\ 
 \hline
 Merged & 1.000 &  1.000 &  1.000 &  1.000 &  1.000 & 0.995 & 0.965 & 0.805 & 0.550 \\
   Stouffer & 1.000 &  1.000 &    1.000 &  1.000 &  1.000 & 1.000 & 0.990 & 0.905 & 0.745 \\ 
 Fisher & 1.000 &  1.000 &  1.000 &  1.000 & 1.000 & 1.000 & 1.000 & 1.000 & 0.995\\ 
 Truncated $\tilde{\alpha} = 0.05$ & 1.000 & 1.000 & 1.000 & 1.000 & 1.000 &  1.000 &  1.000 &  1.000 &  1.000 \\  
 Bonferroni & 1.000 & 1.000 & 1.000 & 1.000 & 1.000 & 1.000 & 1.000 & 1.000 & 1.000\\
 \hline 
\end{tabular}}
\end{table}

\subsection{Simulation Result}

Table \ref{table:5} summarizes the simulation results for different choices of methods for combining $P$-values, combinations of average treatment effect for two datasets, and sensitivity parameter $\Gamma$ for hidden bias. The corresponding sensitivity parameter for selection bias is $\Theta = 1.0$. Additional simulation results for $\Theta = 1.1$ and $\Theta = 1.2$ can be found in the Appendix VII.  

We identify several consistent patterns from the simulation results. First, we find that when the treatment effects between two datasets are equal, our proposed methodology does about as well as the standard hypothesis test while our proposed methodology does substantially better than the standard as the difference in the average treatment effects between the two datasets increases. Furthermore, among the methods of combining $P$-values, the unweighted Stouffer's $Z$-score method demonstrates the lowest power which illustrates that the unweighted Stouffer's $Z$-score could be somewhat conservative under certain types of settings \citep{whitlock2005combining}. The Bonferroni method typically delivers a test that is as powerful as and in many cases more powerful than other type of methods of combining $P$-values. Generally Bonferroni is not very conservative when it is applied two independent tests.  Hence, we recommend using our proposed methodology with the Bonferroni method when there are two independent types of cases for which the treatment effect might differ.

\section{Case Study:  Alcohol Consumption and Risk of Postmenopausal Invasive Breast
Cancer by Subtypes}

\label{sec: case-studies}

Returning to the motivating example, of the 4,046 invasive breast cancer patients in the data, there are 695 patients who drank alcohol 18+ oz per week. And of the 74,993 non-invasive breast cancer patients, there are 1,495 patients who drank alcohol 18+ oz per week. We matched each of the 4,046 invasive breast cancer patients (case) to one non-invasive breast cancer patient (non-case) on age, body mass index, ethnicity, family income, blood pressure, education, hormone replacement therapy (HRT) usage, and smoking status. The matching process is based on \citet{zhang2021matching} by implementing \textsf{R} package \textsf{match2C}. We minimize the Mahalanobis distance in the matching algorithm. See Chapter 8 of \citet{rosenbaum2010design} for a detailed discussion. Our matched data include missing information on ethnicity (13 out of 4,046 cases, and 11 out of 4,046 non-cases), family income (180 out of 4,046 cases, and 178 out of 4,046 non-cases), education (28 out of 4,046 cases, and 28 out of 4,046 non-cases), and smoking status (1792 out of 4,046 cases, and 1788 out of 4,046 non-cases). We followed \citet{rosenbaum2010design} in handling the missingness by for categorical variables including a category for missingness and for continuous variables (i) plugging in the mean value when missing and (ii) including an indicator variable for missingness.

 The distribution of demographic and personal characteristics stratified by invasive breast cancer as well as the balance of matching can be found in Table \ref{table:6}. Before matching, some covariates such as ethnicity and body mass index have absolute standardized mean differences greater than 0.2. After matching, the absolute standardized mean differences of every covariate are less than 0.08. We also assessed the matching quality through the Sample-Splitting Classification Permutation Test by a logistic regression classifier, as outlined in \citep{chen2023testing}. We did not reject the null hypothesis that the balance of the measured variables is what we would expect in a randomized experiment ($P$-value = 0.21).  Comparing women who drank 18+ oz per week versus those who do not, the odds ratio for invasive breast cancer is 1.58 with 95\% confidence interval $[1.16, 2.16]$ (see Table \ref{table:7} in the supplementary material).  Under the assumption of no unmeasured confounding, a $95\%$ one-sided prediction interval for the attributable effect $A$ is $A \geq 17$. This prediction interval has the property that if treatment were assigned at random in each matched pair (conditional on the number of treated units in each matched pair), then in at least $95 \%$ of repeated studies, the number of invasive cases that are attributed to drinking 18+ oz per week of alcohol would be in the interval. Consequently, assuming no unmeasured confounding, we can have some ``confidence'' that at least $17/103 = 16.5 \%$ of the invasive breast cancer cases are attributable to drinking 18+ oz per week of alcohol. The $P$-value for testing the Fisher's null hypothesis $H_0: \bm \delta = \bm 0$ is $0.002$ under the assumption of no unmeasured confounding. This evidence that alcohol increases invasive breast cancer is robust to hidden bias of $\Gamma = 1.22$, meaning that two matched subjects differing in their odds of 18+ oz per week of alcohol intake by at most a factor of $1.22$ would not be able to explain away the apparent effect of drinking 18+ oz per week of alcohol on invasive breast cancer.


 Leveraging the case description information may potentially strengthen the causal conclusion to be more robust against unmeasured confounding because the effect of 18+ oz per week of alcohol intake is more noticeable among hormone sensitive breast cancer in comparison to hormone insensitive breast cancer. Out of 4,046 invasive breast cancer cases, 3,177 cases are hormone sensitive while 869 cases are hormone insensitive. The odds ratio of having hormone sensitive and insensitive invasive breast cancer are 2.00 (95\% CI $[1.39, 2.88]$) and 0.71 (95\% CI $[0.37, 1.38]$), respectively. We apply our developed methodology to the matched data. We first set $\Theta = 1.00$, which means drinking alcohol 18+ oz per week would not affect the status of an invasive breast cancer patient (hormone sensitive or insensitive). Under this setting, one would observe that Fisher's method would be insensitive to a bias of $\Gamma = 1.30$ on rejection of $H_0: \bm \delta = \bm 0$,   and the Stouffer's $Z$-score method and truncated Fisher's method with $\tilde{\alpha} = 0.10$ would be insensitive to a bias of $\Gamma = 1.38$ while the Bonferroni method would be insensitive to a bias of $\Gamma = 1.40$. These results indicate the significant improvement in robust to hidden bias by using case description information (without using hidden bias, the results were only insensitive up to $\Gamma = 1.22$).


We now consider $\Theta = 1.10$, which means that among all invasive breast cancer patients, up to $10\%$ of them who have hormone insensitive breast cancer would have hormone sensitive breast cancer if they were drinking alcohol 18+ oz per week. Under this setting, Fisher's method would be insensitive to a bias of $\Gamma = 1.18$ on rejection of $H_0: \bm \delta = \bm 0$, and the Stouffer's $Z$-score method and the truncated Fisher's method with $\tilde{\alpha} = 0.10$ would be insensitive to a bias of $\Gamma = 1.26$ while the Bonferroni method would be insensitive to a bias of $\Gamma = 1.28$. We observed that our proposed methodology is more robust to hidden bias even when selection bias is introduced compared with the method without using case description information.

\begin{table}
\hspace*{-1cm}
\caption{\label{table:6} \small Balance table of all samples and matched samples as well as the distribution of demographic and personal characteristic stratified by invasive breast cancer. There are 4,046 matched pairs.}
\centering
\resizebox{1\textwidth}{!}{
\begin{tabular}{lccccc}
  \hline \multirow{3}{*}{\begin{tabular}{c}\textbf{Characteristic}\end{tabular}}
   & \multirow{3}{*}{\begin{tabular}{c}Non-invasive breast cancer \\Non-case \\  ($n = 74993$ All controls)\end{tabular}} &
   \multirow{3}{*}{\begin{tabular}{c}Non-invasive breast cancer \\Non-case \\ ($n = 4046$ Matched controls)\end{tabular}} &
   \multirow{3}{*}{\begin{tabular}{c}Invasive breast cancer \\Case  \\  ($n = 4046$ )\end{tabular}} &  \multirow{3}{*}{\begin{tabular}{c}Standardized \\mean difference \\ (Before matching) \end{tabular}}   & \multirow{3}{*}{\begin{tabular}{c}Standardized \\mean difference \\ (After matching) \end{tabular}}\\ \\ \\
  \hline
    \textbf{Breast cancer, No.(\%)} \\
  \hspace{0.3cm} No & 74132(98.9) & 3969(98.1) &  0(0.0) & & \\ 
  \hspace{0.3cm} Yes & 861(1.1)	 & 77(1.9) & 4046(100.0) & &  \\ 
   \textbf{Hormone sensitive breast cancer, No.(\%)} \\
  \hspace{0.3cm} No & 556(0.8) & 58(1.5) &  869(21.0) & & \\ 
  \hspace{0.3cm} Yes & 303(0.4)	 & 19(0.5) & 3177(79.0) & &  \\ 
  \hspace{0.3cm} Missing & 74134(98.8)	 & 3969(98.0) & 0(0.0) & &  \\
     \textbf{Alcohol intake 18+ oz a week, No.(\%)} \\
  \hspace{0.3cm} No & 73498(98.0) & 3980(98.4) &  3943(97.5) & & \\ 
  \hspace{0.3cm} Yes & 1495(2.0)	 & 66(1.6) & 103(2.5) & &  \\ 
  \textbf{Age at enrollment} & 63(57, 69) & 64(58,69) & 64(58,69)	 & 0.04 & 0.00\\
  \textbf{Body Mass Index (BMI)} & 26.0(23.2,29.9) & 26.2(23.4,29.7) & 26.2(23.3,30.0) & 0.21 & 0.07 \\
  \textbf{Ethnicity} \\
  \hspace{0.3cm} American Indian or Alaskan Native & 304(0.4) & 11(0.3) & 13(0.3) & 0.04 & 0.01 \\ 
  \hspace{0.3cm} Asian or Pacific Islander & 1447(1.9) & 51(1.3) & 50(1.2) & 0.06 & -0.01 \\ 
  \hspace{0.3cm} Black or African American & 5681(7.6) & 224(5.5) & 224(5.5) & 0.23 & 0.04 \\ 
  \hspace{0.3cm} Hispanic/Latino & 2524(3.4) & 78(1.9) & 78(1.9) & 0.08 & 0.00 \\ 
  \hspace{0.3cm} White (not of Hispanic origin) & 64035(85.0) & 3641(90.0) & 3641(90.0) & -0.26 & -0.03 \\ 
  \hspace{0.3cm} Other & 804(1.1) & 28(0.7) & 29(0.7) & 0.02 & 0.00 \\ 
  \hspace{0.3cm} Missing & 198(0.3) & 13(0.3) & 11(0.3) & 0.03 & 0.00 \\ 
  \textbf{Family Income} \\
  \hspace{0.3cm} Less than \$10000 & 2804(3.7) & 90(2.2) & 92(2.3) & 0.18 & 0.06 \\ 
  \hspace{0.3cm} \$10000 to \$19999 & 7570(10.0) & 342(8.5) & 342(8.5) & 0.18 & 0.03 \\ 
  \hspace{0.3cm} \$20000 to \$34999 & 15987(21.0) & 850(21.0) & 845(21.0) & 0.09 & -0.01 \\ 
  \hspace{0.3cm} \$35000 to \$49999 & 14125(18.8) & 792(20.0) & 793(20.0) & -0.04 & -0.02 \\ 
  \hspace{0.3cm} \$50000 to \$74999 & 14329(19.1) & 827(20.0) & 823(20.0) & -0.11 & -0.02 \\ 
  \hspace{0.3cm} \$75000 to \$99999 & 6847(9.1) & 411(10.0) & 412(10.0) & -0.12 & -0.01 \\ 
  \hspace{0.3cm} \$100000 to \$149999 & 5119(6.8) & 290(7.2) & 292(7.2) & -0.13 & -0.01 \\ 
  \hspace{0.3cm} \$150000 or more & 2959(3.9) & 166(4.1) & 164(4.1) & -0.12 & 0.00 \\ 
  \hspace{0.3cm} Don't know & 2027(2.7) & 98(2.4) & 105(2.6) & 0.04 & 0.01 \\ 
  \hspace{0.3cm} Missing & 3226(4.3) & 180(4.4) & 178(4.4) & 0.03 & 0.01 \\ 
  \textbf{Blood Pressure} \\
  \hspace{0.3cm} Systolic & 126(114,138) & 126(114,138) & 126(114,138) & 0.10 & 0.03 \\ 
  \hspace{0.3cm} Diastolic & 74(70,80) & 74(70,80) & 74(70,80) & 0.01 & 0.01 \\ 
  \textbf{Education, No.(\%)} \\
  \hspace{0.3cm} Didn't go to school & 52($<0.1$) & 1($<0.1$) & 2($<0.1$) & 0.01 & 0.00 \\ 
  \hspace{0.3cm} Grade school (1-4 years) & 188(0.3) & 7(0.2) & 7(0.2) & 0.06 & 0.03 \\ 
  \hspace{0.3cm} Grade school (5-8 years) & 671(0.9) & 10(0.2) & 10(0.2) & 0.11 & 0.05 \\ 
  \hspace{0.3cm} Some high school (9-11 years) & 2404(3.2) & 89(2.2) & 88(2.2) & 0.15 & 0.01\\ 
  \hspace{0.3cm} High school diploma or GED & 11624(16.0) & 548(14.0) & 546(13.0) & 0.10 & 0.01 \\ 
  \hspace{0.3cm} Vocational or training school & 7067(9.4) & 349(8.6) & 348(8.6) & 0.08 & 0.00 \\ 
  \hspace{0.3cm} Some college or Associate Degree & 20358(27.0) & 1061(26.0)	 & 1061(26.0) & 0.00 & 0.00 \\ 
  \hspace{0.3cm}College graduate or Baccalaureate Degree & 8702(11.6) & 517(13.0) & 518(13.0) & -0.10 & 0.00 \\ 
  \hspace{0.3cm} Some post-graduate or professional & 9096(12.4) & 579(14.0) & 577(14.0) & -0.09 & -0.01  \\ 
  \hspace{0.3cm} Master's Degree & 12105(16.0) & 741(18.0) & 745(18.0) & -0.11 & -0.02  \\ 
  \hspace{0.3cm} Doctoral Degree (Ph.D,M.D.,J.D.,etc.) & 2161(2.9) & 116(2.9) & 116(2.9)  & -0.06 & -0.01 \\ 
  \hspace{0.3cm} Missing & 592(0.8) & 28(0.7) & 28(0.7) & 0.02 & 0.01  \\ 
  \textbf{HRT usage ever, No. (\%)} \\
  \hspace{0.3cm} No & 29717(40.0) & 1407(34.8) & 1433(35.5) & 0.04 & 0.01 \\ 
  \hspace{0.3cm} Yes & 45276(60.0) & 2639(65.2) & 2610(64.5) & 0.04 & -0.01 \\ 
  \textbf{Current smoking status, No.(\%)} \\
   \hspace{0.3cm} Nonsmoker & 34784(46.3) & 2015(50.0) & 1987(49.0) & -0.03 & 0.01 \\ 
  \hspace{0.3cm} Current smoker & 5076(6.8) & 239(5.8) & 271(6.8) & 0.00 & 0.00 \\
  \hspace{0.3cm} Missing & 35133(46.9)	 & 1792(44.2) & 1788(44.2) & -0.06 & 0.00 \\
   \hline
\end{tabular}}
\end{table}

\begin{table}
\hspace*{-2cm}
\caption{\label{table:10} Various sensitivity analyses for statistical inference on the attributable fraction among the exposed ($\text{AF}_e$). The last five columns report the $95\%$ minimum prediction interval of $\text{AF}_e$ for different methods in different combinations of $\Gamma$ and $\Theta$. The bold number of $\Gamma$ and $\Theta$ represents the maximum magnitude of sensitivity parameters that the rejection of $H_0: \bm \delta = \bm 0$ is insensitive for corresponding methods.  }
\centering
\hskip-2.0cm
\begin{tabular}{c  c c c c c c c c c} 
 \hline
 \multicolumn{10}{c}{$95\%$ minimum prediction interval for $\text{AF}_e$}   \\
 \hline 
 $\Gamma$ & $\Theta$ &    &  &  & Merged & Stouffer  & Fisher  & Truncated $\tilde{\alpha} = 0.10$  & Bonferroni     \\ 
 \hline
 $1.00$ & $1.00$ &  & &  &$\text{AF}_e \geq 16.50 \%$ & $\text{AF}_e \geq 15.53 \%$ &$\text{AF}_e \geq 18.45\%$  & $\text{AF}_e \geq 22.33\%$  & $\text{AF}_e \geq 22.33\%$ \\ 
 $1.08$ & $1.00$ & & & & $\text{AF}_e \geq 10.68 \%$ & $\text{AF}_e \geq 12.62 \%$ &$\text{AF}_e \geq 13.59 \%$  & $\text{AF}_e \geq  17.48 \%$ & $\text{AF}_e \geq 18.45 \%$  \\ 
 $1.16$ & $1.00$ &  & & & $\text{AF}_e \geq 4.85 \% $ & $\text{AF}_e \geq 8.74 \%$ &$\text{AF}_e \geq 8.74 \%$  & $\text{AF}_e \geq 12.62 \%$  & $\text{AF}_e \geq 13.59 \%$ \\ 
$\mathbf{1.22}$ & $\mathbf{1.00} $ &  & & & $\textbf{AF}_e \mathbf{\geq 0.00 \%} $ & $\text{AF}_e \geq 6.80 \%$ &$\text{AF}_e \geq 4.85 \%$  & $\text{AF}_e \geq 9.71 \%$  & $\text{AF}_e \geq 9.71 \%$ \\ 
$1.26$ & $1.00$ &  & & & $\text{AF}_e \geq 0.00 \%$ & $\text{AF}_e \geq 4.85 \%$ &$\text{AF}_e \geq  1.94 \%$  & $\text{AF}_e \geq 6.80 \%$   & $\text{AF}_e \geq 7.77 \%$\\ 
$\textbf{1.30}$ & $\textbf{1.00}$ &  & & & $\text{AF}_e \geq 0.00 \%$ & $\text{AF}_e \geq 2.91 \%$ &$ \textbf{AF}_e \mathbf{\geq 0.00 \%}$  &  $\text{AF}_e \geq 4.85 \%$   & $\text{AF}_e \geq 5.83 \%$\\ 
$1.34$ & $1.00$ &  & & & $\text{AF}_e \geq 0.00 \%$ & $\text{AF}_e \geq 1.94 \%$ &$\text{AF}_e \geq  0.00 \%$  & $\text{AF}_e \geq 2.91 \%$   & $\text{AF}_e \geq 2.91 \%$\\ 
$\textbf{1.38}$ & $\textbf{1.00}$ &  & & & $\text{AF}_e \geq 0.00 \%$ & $\textbf{AF}_e \mathbf{\geq 0.00 \%}$ &$\text{AF}_e \geq 0.00 \%$  &  $\textbf{AF}_e \mathbf{\geq 0.00 \%}$   & $\text{AF}_e \geq 0.97 \%$\\ 
$\textbf{1.40}$ & $\textbf{1.00}$ &  & & & $\text{AF}_e \geq 0.00 \%$ & $\text{AF}_e \geq 0.00 \%$ &$\text{AF}_e \geq 0.00 \%$  &  $\text{AF}_e \geq 0.00 \%$   & $\textbf{AF}_e \mathbf{\geq 0.00 \%}$\\ 
 $1.04$ & $1.10$ &  & & & $\text{AF}_e \geq 5.83 \%$ & $\text{AF}_e \geq 9.71 \%$ &$\text{AF}_e \geq 9.71 \%$  & $\text{AF}_e \geq 13.59 \%$   & $\text{AF}_e \geq 14.56 \%$\\ 
$1.08$ & $1.10$ & & & & $\text{AF}_e \geq 2.91 \%$ & $\text{AF}_e \geq 7.77 \%$ &$\text{AF}_e \geq 6.80 \%$  & $\text{AF}_e \geq 11.65 \%$   & $\text{AF}_e \geq 11.65 \%$\\ 
$\textbf{1.12}$ & $\textbf{1.10}$ & & & & $\textbf{AF}_e \mathbf{\geq 0.00 \%}$ & $\text{AF}_e \geq 5.85 \% $ &$\text{AF}_e \geq 3.88 \%$  & $\text{AF}_e \geq 8.74 \%$  & $\text{AF}_e \geq 9.71 \%$ \\ 
$\textbf{1.18}$ & $\textbf{1.10}$ & & & & $\text{AF}_e \geq 0.00 \%$ & $\text{AF}_e \geq 3.88 \%$ &$\textbf{AF}_e \mathbf{\geq 0.00 \%}$  &  $\text{AF}_e \geq 4.85 \%$   & $\text{AF}_e \geq 5.83 \%$ \\ 
$\textbf{1.26}$ & $\textbf{1.10}$ & & & & $\text{AF}_e \geq 0.00 \%$ & $\textbf{AF}_e \mathbf{\geq 0.00 \%}$ &$\text{AF}_e \geq 0.00 \%$  &  $\textbf{AF}_e \mathbf{\geq 0.00 \%}$   & $\text{AF}_e \geq 0.97 \%$ \\ 
$\textbf{1.28}$ & $\textbf{1.10}$ & & & & $\text{AF}_e \geq 0.00 \%$ & $\text{AF}_e \geq 0.00 \%$ &$\text{AF}_e \geq 0.00 \%$  &  $\text{AF}_e \geq 0.00 \%$   & $\textbf{AF}_e \mathbf{\geq 0.00 \%}$ \\ 
 \hline
\end{tabular}
\end{table}

\section{Discussion: Summary and Extension}
\label{sec: conclusion}

Often there are distinctions between cases, and some types of cases are more responsive to the exposure than others in disease analysis. This case definition information can be related to morphologic, immunophenotypic, genetic, and clinical features of cases. For example, non-Hodgkin's lymphoma is a group of illnesses that is characterized by the malignant transformation of lymphoid cells. To understand patterns of etiologic commonality and heterogeneity for non-Hodgkin's lymphomas, \citet{morton2008etiologic} presented the first systematic comparison of risks by lymphoma subtype for a broad range of putative risk factors in a population-based case-referent study. Their study found that patients have higher risk for some non-Hodgkin lymphoma subtypes like marginal zone lymphomas than subtypes like chronic lymphocytic leukemia/small lymphocytic lymphoma when the patients were exposed to some autoimmune conditions (more detailed information can be found in Table 1-4 in \citet{morton2008etiologic}).

In this paper, we first propose a methodology to use case definition information for robust statistical inference on the attributable fraction among the exposed in terms of unmeasured confounding. Then to tackle the selection bias issue, we adapt a new sensitivity analysis framework that includes sensitivity parameters for both selection bias and unmeasured confounding in matched case-referent studies.  In our developed methodology, we implement the Stouffer's $Z$-score method, the Fisher type of combining $P$-values method, and the Bonferroni method. We study comprehensively the power of sensitivity analysis and design sensitivity by different settings of simulation. In our application, we find that implementing our methodology of using case description information, the inference on the attributable fraction among the exposed is more robust to unmeasured confounding compared with the conclusion without using the case definition information. 

In practice, we can leverage the case definition information when (i) it offers a larger effect size, and (ii) the proportion of one subtype of cases among all cases that changed to another subtype under the influence of the exposure can be controlled by the sensitivity parameter $\Theta$ for selection bias. The proposed methodology and sensitivity analysis framework are based on the inference of the attributable effect in matched case-referent studies. It can further be extended to different causal design settings such as matched cohort studies. Additionally, post-matching adjustment for covariates in the inference for the attributable risk could be another extension. These are possible future directions for the extension of our framework.

\section*{Acknowledgment}
This research was partially supported by NIH RF1AG063481 grant.

\bibliographystyle{apalike}
\bibliography{reference}

\newpage

\begin{center}
    \Large Supplementary Materials for ``Using Case Description Information to Reduce Sensitivity to Bias for the Attributable Fraction Among the Exposed''
\end{center}

\section*{Appendix I: Confidence Set of $\bm \delta$ and Attributable Effects}

Table \ref{tab: AE and CS} represents a $2 \times 2$ contingency table where the count in the upper-left cell is reduced by $a$, while the count in the lower-left cell is correspondingly increased by $a$, resulting in adjusted row totals. The observed table corresponds to the special case where $a = 0$. If $D-a$ in Table \ref{tab: AE and CS} has the extended hypergeometric distribution, then $\mathbb{P}(D - a = b)$ is $\pi(a,b)$ given by 

\begin{align*}
    \pi(a,b) = \left\{ \binom{L-a}{b} \binom{N-L+a}{m-b} \Gamma^b   \right\} \Bigg/ \left\{ \sum_{c = \max(0,m+L-a-N) }^{\min(m,L-a)}  \binom{L-a}{c} \binom{N-L+a}{m-c} \Gamma^c   \right\},
\end{align*}
if $\min(m,L)\geq b \geq \max(0,m+L-a-N)$, and is given by $\pi(a,b) = 0$ otherwise; and $\mathbb{P}(D-a\geq k)$ is $\lambda(a,k) = \sum_{b=\max(k,m+L-a-N)}^{\min(m,L-a)}\pi(a,b)$. 

In drawing inferences, various values of $a$ are considered. Proposition \ref{prop: confidence set} compares the tail probability $\lambda(a, k)$ across these values of $a$. It is important to note that this comparison is unconventional, as it involves distributions on $2 \times 2$ tables with differing marginal totals. Proposition \ref{prop: confidence set} underpins the assertion that the set of $\bm{\delta}_0$ values not rejected corresponds to an interval of values for the attributable effect, $A_0 = \sum_{i,j} Z_{ij} \delta_{0_{ij}}$. Specifically, Proposition \ref{prop: confidence set} examines the probability that $D-a \geq k$ for two values of $a$. However, it should be emphasized that changes in $a$ alter both the margins of Table \ref{tab: AE and CS} and the associated extended hypergeometric distribution of $D$.

\begin{proposition}[\citet{rosenbaum2001effects}]\label{prop: confidence set}
    If $a \geq a^*$, then $\lambda(a^*,k) \geq \lambda(a,k)$. 
\end{proposition}

\begin{proof}
    It suffices to prove this for $a^* = a - 1$ as the general result follows by induction. Moreover, it suffices to prove that if $b < b^*$ then 
    \begin{align*}
        \pi(a,b) \pi(a-1, b^*) \geq \pi(a,b^*) \pi(a-1, b),
    \end{align*}
by the familiar fact that the ordering implies stochastic ordering \citep{eaton1987lectures}. There are two cases to be considered. First, it is straightforward to verify that if $\pi(a,b)\pi(a-1,b^*) = 0$, then $\pi(a,b^*)\pi(a-1,b)=0$ also. Secondly, if $\pi(a,b^*)\pi(a-1,b)>0$ then $\pi(a,b)\pi(a-1,b^*)$ and $\pi(a,b^*)\pi(a-1,b)$ have the same denominator, and simple algebra applied to the numerators suffices to complete the proof. 
\end{proof}

\begin{table}
\hspace{-5cm}
\label{tab: AE and CS}
\caption{ A $2 \times 2$ table adjusted for attributable effects }
\centering
\begin{tabular}{|| c c c c||}
\hline
 Response & Treated  & Control  & Total   \\
 \hline\hline
1 & $D-a$ & $L - D$  & $L - a$  \\ 
2 & $m - D + a$ & $N - m - L + D$ & $N - L + a$  \\  
Total & $m$ & $N-m$    & $N$  \\
 \hline
\end{tabular}
\end{table}

\newpage

\section*{Appendix II: Justification of Approximation Using Asymptotic Separability}

A key question is whether it is plausible that no more than $a$ treated events attributable to treatment. Formally, this is equivalent to assessing whether the number of units for whom the treatment effect is $\delta_{ij} = r_{T_{ij}} - r_{C_{ij}} = 1$ satisfies $A = \sum_{i,j}Z_{ij}\delta_{ij}\leq a$. The goal is to determine whether there exists $\bm \delta_0 \in \mathfrak{S}$ such that $A =\sum_{i,j}Z_{ij}\delta_{ij} \leq a $ where $\mathfrak{S}$ is the confidence set for $\bm \delta$. Alternatively, if no $\bm \delta_0$ in the confidence set $\mathfrak{S}$ satisfies this criterion, we can conclude that $A \leq a$ is implausible. 

Three approaches are proposed to address this question. The first approach provides exact randomization inference and is theoretically rigorous but computationally intensive. The second approach offers a large-sample approximation based on the normal distribution, but is nearly as tedious as the first approach. The third approach employs asymptotic separability, a practical method that allows a single $\bm \delta_0$ to be examined to infer properties of the entire confidence set. This approach is particularly advantageous due to its simplicity and efficiency.

In the exact inference, one evaluates whether there exists a $\bm \delta_0 \in \mathfrak{S}$ for which $A = \sum_{i,j}Z_{ij}\delta_{0_{ij}} \leq a$. If no such $\bm \delta_0$ exists, every configuration implying $A \leq a$ is rejected as implausible. However, exact inference requires explicit computation of the confidence set $\mathfrak{S}$ and verification of all possible $\bm \delta_0 \in \mathfrak{S}$, which can be tedious and computationally intensive. The second approach is to use a large sample approximation based on the normal distribution to test each $\bm \delta_0$. But we still need to verify all possible $\bm \delta_0 \in \mathfrak{S}$ which can be computational intensive as well.

We now discuss the third approach that uses asymptotic separability and is computationally efficient. We present how to check a single $\bm \delta_0$ and on the basis of that one $\bm \delta_0$ reach a conclusion about all of $\mathfrak{S}$ instead of checking every $\bm \delta_0 \in \mathfrak{S}$ by using the notion of asymptotic separability \citep{gastwirth2000asymptotic} where $\mathfrak{S}$ is the confidence set for $\bm \delta$. The approximation will find one $\bm \delta_0$ with $a = \sum_{i,j}Z_{ij}\delta_{0_{ij}}$ that is hardest to reject; it then must be checked whether this $\bm \delta_0$ is in $\mathfrak{S}$. We call this $\bm \delta_0$ the ``worst case'' which is the one that yields the highest $P$-value in the set of $\bm \delta_0$ that are compatible with $a$ when testing with sensitivity parameter $\Gamma$. If this $\bm \delta_0 \in \mathfrak{S}$, then it is plausible that $a \geq \sum_{i,j} Z_{ij} \delta_{0_{ij}}$, and it is plausible that $a$ or fewer events among treated subjects were caused by treatment; however, if this $\bm \delta_0 \notin \mathfrak{S}$, then this $\bm \delta_0$ and all others with $a \geq \sum_{i,j} Z_{ij} \delta_{0_{ij}}$ have been rejected, so it is not plausible that $a$ or fewer events are attributable to treatment. This is a large-sample approximation based on a limiting normal distribution characterized by its expectation and variance.

The approximation will construct one $\bm \delta_0$ such that in $a$ of the matched sets with $\sum_{j=1}^J Z_{ij}\delta_{0_{ij}} = 1$, the event was indeed caused by the treatment, so that $a = \sum_{i,j} Z_{ij}\delta_{0_{ij}}$. The approximation considers each matched set, one at a time, and determines by how much the maximum expected contribution $\tilde{\tilde{\pi}}_i$ from this set would decline if it were assumed that a treated subject in this set had an event caused by the treatment. The decline is given by $\bar{\bar{\lambda}}_i - \bar{\lambda}_i$ ($\bar{\bar{\lambda}}_i$ and $\bar{\lambda}_i$ are defined in the Appendix III), where $\bar{\bar{\lambda}}_i$ is the value of $ \tilde{\tilde{\pi}}_i$ if the event in matched set $i$ was not attributable so that $\sum_{j=1}^J Z_{ij}\delta_{0_{ij}} = 0$ and $R_{i^+} = r_{C{i^+}}$ whereas $\bar{\lambda}_i$ is the value of $ \tilde{\tilde{\pi}}_i$ if the event in matched set $i$ was attributable so that $\sum_{j=1}^J Z_{ij}\delta_{0_{ij}} = 1$ and $R_{i^+} = r_{C_{i^+}} + 1$.

The $\bm \delta_0$ formed using the $a$ smallest declines, $\bar{\bar{\lambda}}_i - \bar{\lambda}_i$, has the highest expected number of events among treated subjects. Ties are broken by picking the sets that do the least to decrease the variance, the decrease being $\bar{\bar{w}}_i - \bar{w}_i$ ($\bar{\bar{w}}_i$ and $\bar{w}_i$ are defined in the Appendix II). The following proposition shows that as $I \rightarrow \infty$, the largest approximation significance level (\ref{equ: approx upper}) is obtained from the normal reference distribution with the highest expectation, and where there are several normal distributions with the same highest expectation, the one among these with the highest variance. A normal distribution with a high expectation and variance attaches high probabilities to large values. Asymptotic separability is an approximation that works as $I \rightarrow \infty$.

Denote $\mathbb{P}(B_i = 1) = \hat \pi_i$, $\hat{\sigma}_I = \sqrt{\sum_i \hat \pi_i (1-\hat \pi_i) }$ and $\tilde{\tilde{\sigma}}_I = \sqrt{\sum_i \tilde{\tilde{\pi}}_i (1-\tilde{\tilde{\pi}}_i) } $. We know that $ \tilde{\tilde{\pi}}_i \geq \hat \pi_i$ by Section 2.3. And if $ \tilde{\tilde{\pi}}_i = \hat \pi_i$, then $\tilde{\tilde{\pi}}_i (1 - \tilde{\tilde{\pi}}_i) \geq \hat \pi_i (1- \hat \pi_i) $ for all $i$. Define $\Delta = \{\tilde{\tilde{\pi}}_i - \hat \pi_i: i = 1,\cdots, I  \}$, and let $\delta = \inf(\Delta)$. Fix $k > 0$, $\hat k_I = k \hat \sigma_I + \sum_i \hat \pi_i$,  $\tilde{\tilde{S}}_I = \sum_i Z_{ij} r_{T_{ij}} $, $\hat S_I =\sum_i Z_{ij} r_{C_{ij}} =\sum_i B_i$. Note that by the assumption of non-negative treatment effect, we have $r_{T_{ij}} \geq r_{C_{ij}}$, therefore $\tilde{\tilde{S}}_I \geq \hat S_I$. 

\begin{proposition}[Asymptotic Separability \citet{gastwirth2000asymptotic}]
Assuming that there are two numbers, $0 < \bar \nu^2 \leq \bar{\bar{\nu}}^2 < \infty$ so that $\bar \nu^2 \leq \tilde{\tilde{\pi}}_i (1 - \tilde{\tilde{\pi}}_i)  \leq \bar{\bar{\nu}}^2 $ and  $\bar \nu^2 \leq \hat \pi_i (1 - \hat \pi_i)  \leq \bar{\bar{\nu}}^2 $ for all $i$. If $\delta > 0$, then for all $\epsilon > 0$ there is an $I^*$ such that $\mathbb{P} (\tilde{\tilde{S}}_I \geq \hat k_I   ) \geq \mathbb{P} (\hat S_I \geq \hat k_I   ) - \epsilon$ for $I \geq I^*$. 
\end{proposition}

\begin{proof}
By the central limit theorem, for $\epsilon >0$ there is an $I_1$ such that for $I \geq I_1$ both
\begin{align*}
    \mathbb{P}\{ (\tilde{\tilde{S}}_I - \sum_i \tilde{\tilde{\pi}}_i  )/\tilde{\tilde{\sigma}}_I \geq k   \} \geq 1 - \Phi(k) - \epsilon/2
\end{align*}
and 
\begin{align*}
    \mathbb{P}\{ (\hat S_I - \sum_i \hat \pi_i  )/\hat \sigma_I \geq k   \} \geq 1 - \Phi(k) + \epsilon/2. 
\end{align*}
By simple algebra, 
\begin{align*}
    \mathbb{P}(\tilde{\tilde{S}}_I \geq \hat k_I ) = \mathbb{P} \left\{  \frac{\tilde{\tilde{S}}_I - \sum_i \tilde{\tilde{\pi}}_i }{ \tilde{\tilde{\sigma}}_I }   \geq \frac{\sum_i (\hat \pi_i - \tilde{\tilde{\pi}}_i ) + k \hat \sigma_I  ) }{\tilde{\tilde{\sigma}}_I }         \right\}.
\end{align*}
If there were an $I_2$ such that, for all $I \geq I_2$,
\begin{align*}
    k \geq \frac{\sum_i (\hat \pi_i - \tilde{\tilde{\pi}}_i ) + k \hat \sigma_I  ) }{\tilde{\tilde{\sigma}}_I },    
\end{align*}
then, for all $I \geq \max(I_1, I_2)$, the probability $  \mathbb{P}(\tilde{\tilde{S}}_I \geq \hat k_I )$ would be greater than or equal to $1 - \Phi(k) - \epsilon/2$, and $\mathbb{P} (\hat S_I \geq \hat k_I   ) $ would be less than $1 - \Phi(k) + \epsilon/2$, so the proof would be complete. We now show that such an $I_2$ exists. Equivalently, we show that there is an $I_2$ such that, for all $I \geq I_2$, 
\begin{equation} \label{equ:ineq}
        \frac{1}{I} \sum_i \tilde{\tilde{\pi}}_i - \hat \pi_i \geq \frac{k}{\sqrt{I}} \left\{   \sqrt{\frac{1}{I} \sum_i \hat{\pi}_i (1-\hat{\pi}_i)  } - \sqrt{ \frac{1}{I} \sum_i \tilde{\tilde{\pi}}_i (1-\tilde{\tilde{\pi}}_i) }   \right\}.  
\end{equation}
Let $A_I = \{i,  \tilde{\tilde{\pi}}_i (1-\tilde{\tilde{\pi}}_i) - \hat{\pi}_i (1-\hat{\pi}_i)  < 0  \}$ and let $O_I = |A_I|/I$, where the number of elements in a set $V$ is denoted $|V|$. If $i \in A_I$ then $\tilde{\tilde{\pi}}_i - \hat \pi_i \geq \delta$, so
\begin{align*}
    \frac{1}{I} \sum_i \tilde{\tilde{\pi}}_i - \hat \pi_i \geq \delta O_I. 
\end{align*}
To complete the proof, it suffices to show that the right-hand side of inequality (\ref{equ:ineq}) is less than $\delta O_I$ for all sufficiently large $I$. Define $\tilde{\tilde{\psi}}_I$ and $\hat \psi_I $ to be
\begin{align*}
    \hat \psi_I = \frac{1}{I - |A_I|} \sum_{i \notin A_I} \hat \pi_i (1- \hat \pi_i)
\end{align*}
and 
\begin{align*}
   \tilde{\tilde{\psi}}_I = \frac{1}{I - |A_I|} \sum_{i \notin A_I} \tilde{\tilde{\pi}}_i (1 - \tilde{\tilde{\pi}}_i),
\end{align*}
if $|A_I| < I$, and to be 0 if $|A_I| = I$. By the definition of $A_I$, if $i\notin A_I$ then $\tilde{\tilde{\pi}}_i (1 - \tilde{\tilde{\pi}}_i) - \hat \pi_i (1- \hat \pi_i) \geq 0 $, so $0 \leq   \hat \psi_I \leq \tilde{\tilde{\psi}}_I$. Also, since $\hat \pi_i (1 - \hat \pi_i) \leq  \bar{\bar{\nu}}^2$ and $\tilde{\tilde{\pi}}_i (1 - \tilde{\tilde{\pi}}_i) \geq  \bar \nu^2$. 
It follows that 

\begin{align*}
   \frac{k}{\sqrt{I}} \left\{   \sqrt{\frac{1}{I} \sum_i \hat{\pi}_i (1-\hat{\pi}_i)  } - \sqrt{ \frac{1}{I} \sum_i \tilde{\tilde{\pi}}_i (1-\tilde{\tilde{\pi}}_i) }   \right\} \leq \frac{k}{\sqrt{I}} \left\{ \sqrt{O_I \bar{\bar{\nu}}^2 + (1 - O_I) \hat \psi_I} -  \sqrt{O_I \bar{\nu}^2 + (1 - O_I) \hat \psi_I}  \right\}
\end{align*}
which is less than or equal to 
\begin{equation} \label{inequ}
        \frac{k}{\sqrt{I}} \frac{O_I }{2  \sqrt{O_I \bar{\nu}^2 + (1 - O_I) \hat \psi_I}  } (\bar{\bar{\nu}}^2 - \bar{\nu}^2  ) \leq \frac{k}{\sqrt{I}}\frac{O_I}{2 \bar{\nu}} (\bar{\bar{\nu}}^2 - \bar{\nu}^2  ),
\end{equation}
because $\sqrt{ax'+b} - \sqrt{ax+b}\leq \frac{a}{2\sqrt{ax+b}}(x'-x)$ for non-negative $a,b,x,x'$ with $ax+b >0$ and $ax'+b > 0$. But the right-hand side of inequality (\ref{inequ}) is a non-negative constant multiple of $O_I/\sqrt{I}$ and so can be made less than or equal to $\delta O_I$ for sufficiently large $I$. 
\end{proof}

\section*{Appendix III: Upper Bound of the Attributable Effect and the General Procedure Finding the Attributable Effect}
 To construct an upper bound on a prediction interval for the attributable effect, we implement an analogous strategy by approximating the lower bound of $\mathbb{P} \left(\sum_{i=1}^I B_i \geq T -A_0 \right) $ using 

\begin{equation}
\label{equ: approx lower}
1 - \Phi \left( \frac{T -A_0 - \sum_{i=1}^I \tilde{\pi}_i }{\sqrt{\sum_{i=1}^I \tilde{\pi}_i (1 - \tilde{\pi}_i)  } }     \right)
\end{equation}
when $I \rightarrow \infty$. Starting at $A_0 = 0$, if (\ref{equ: approx lower}) is less than $1 - \alpha$, $A_0$ should be increased by one until it is not. Assume that $A_0 = a^*$,  then  $\{A: A \leq a^*\}$ is a one-sided $100 \times (1 - \alpha) \% $ maximum prediction interval for the attributable effect $A$.

Next, we demonstrate a general procedure to find the attributable effect. The following steps find the $\bm \delta_0$ with $a = \sum_{i,j} Z_{ij} \delta_{0_{ij}}$ that maximize the expectation $\sum_{i} \tilde{\pi}_i$, and if several $\bm \delta_0$ do this, then finds among these the one that also maximizes the variance $\sum_{i} \tilde{\pi}_i (1-\tilde{\pi}_i)$:

\begin{itemize}
 \item[1.] If $a \geq \sum_{i,j} Z_{ij} R_{ij}$, then $a$ or fewer of the treated subjects had events, whether caused by exposure or not, so it is certain that $a$ or fewer had events caused by the exposure; stop. Otherwise, if $\sum_{i,j}Z_{ij} R_{ij} > a$, then continue 2.
    \item[2.] For each matched set $i$ with $\sum_{j=1}^J Z_{ij}R_{ij} = 1$, calculate
    \begin{align*}
        \bar{\bar{\lambda}}_i &= \frac{\Gamma Z_{i^+}R_{i^+}}{\Gamma Z_{i^+}R_{i^+} + J - Z_{i^+}R_{i^+}} \\
        \bar{\lambda}_i &= \frac{\Gamma Z_{i^+}(R_{i^+}-1)}{\Gamma Z_{i^+}(R_{i^+}-1) + J - Z_{i^+}(R_{i^+}-1)} \\
        \bar{\bar{w}}_i &= \bar{\bar{\lambda}}_i (1 - \bar{\bar{\lambda}}_i) \\
        \bar{w}_i & =  \bar{\lambda}_i (1 -  \bar{\lambda}_i)
    \end{align*}
  \item[3.] Select exactly $a$ of the matched sets $i$ with $\sum_{j=1}^J Z_{ij}R_{ij} = 1$ having the smallest values of $\bar{\bar{\lambda}}_i - \bar{\lambda}_i$. If ties among the $\bar{\bar{\lambda}}_i - \bar{\lambda}_i$ mean that several different groups of $a$ matched sets all have the smallest values of $\bar{\bar{\lambda}}_i - \bar{\lambda}_i$, then among these several groups with the smallest $\bar{\bar{\lambda}}_i - \bar{\lambda}_i$, pick any one group with the smallest values of $\bar{\bar{\lambda}}_i - \bar{\lambda}_i$. For the selected $a$ matched sets, let $\bar{\bar{\pi}}_i = \bar{\lambda}_i$, whereas for the remaining $I - a$ matched sets, let $\bar{\bar{\pi}}_i =  \bar{\bar{\lambda}}_i$. If $a + \sum_{i} \bar{\bar{\pi}}_i \geq \sum_{i,j}Z_{ij}R_{ij}$, then there is one $\bm u$ and one $\bm \delta_0$ with $a = \sum_{i,j} Z_{ij} \delta_{0_{ij}}$ events caused by treatment that would lead us to expect more than the observed number $\sum_{i,j}Z_{ij}R_{ij}$ of events among treated subjects; conclude that this $a$ is plausible and stop; otherwise, continue with step 4.
  \item[4.] Calculate the large-sample approximation to the upper bound on the significance level 
  \begin{align*}
      1 - \Phi (\frac{ (T - a) - \sum_i \bar{\bar{\pi}}_i  }{\sqrt{\sum_{i} \bar{\bar{\pi}}_i (1 - \bar{\bar{\pi}}_i ) }}).
  \end{align*}
  If it is small ($< 0.05$), then conclude that every compatible $\bm \delta_0$ with $a \geq \sum_{i,j} Z_{ij} \delta_{0_{ij}}$ that it is not plausible that exposure to treatment caused $a$ or fewer events. Our setting is a special case of this general procedure in the sense that $R_{i^+} = 1, J = 2$ for every $i$. 
\end{itemize}

\section*{Appendix IV: Proof of Proposition \ref{prop:design sensitivity}}

\begin{proof}[Proof of Proposition \ref{prop:design sensitivity}]
For a fixed $\Theta^*$, if the sensitivity analysis is performed at a $\Gamma > \tilde{\Gamma}$, then all $L$ bounds on $P$-values are tending to $1$ as $I \rightarrow \infty$. So the combined $P$-value is tending to $1$. Let $l'$ be any $l$ such that $\tilde{\Gamma}_{l'} = \tilde{\Gamma}_{\max}$. If the sensitivity analysis is performed with $\Gamma < \tilde{\Gamma}_{\max} = \tilde{\Gamma}$, then the combined $P$-value is tending to $0$ as $I \rightarrow \infty$. So the power of the sensitivity analysis using the combined $P$-value is tending to $1$ when $\Gamma < \tilde{\Gamma}_{\max}$ and tending to $0$ when $\Gamma > \tilde{\Gamma}_{\max}$.

\end{proof}

\section*{Appendix V: Power Calculation of Sensitivity Analysis}

We show the detailed power calculation in the following. Assuming that the null hypothesis is $H_0: \bm \delta = \bm 0$, in this case, the alternative hypothesis is  $H_a: \bm \delta \neq \bm 0$, the test statistic is:
\begin{align*}
     \frac{T  - \sum_{i=1}^I \bar{\bar \pi}_i }{\sqrt{\sum_{i=1}^I \bar{\pi}_i (1 - \bar{\bar \pi}_i)  } }
\end{align*}
where $T = \sum_{i=1}^I C_i$. Since the test statistics under the null hypothesis asymptotically follows a standard normal distribution with the additional assumption that the data is identically distributed. Now suppose that the alternative hypothesis is true and $A_0 = a_*$, then the power is 
\begin{align*}
    B(a_*) & = \mathbb{P}\left( \frac{T - a_* + a_* - \sum_{i=1}^I \bar{\bar \pi}_i }{\sqrt{\sum_{i=1}^I \bar{\pi}_i (1 - \bar{\bar \pi}_i)  } }  > Z_{1 - \alpha} \middle| A_0 = a_* \right) \\
    & = \mathbb{P}\left(\frac{T - a_*  - \sum_{i=1}^I \bar{\bar \pi}_i }{\sqrt{\sum_{i=1}^I \bar{\bar \pi}_i (1 - \bar{\bar \pi}_i)  } } > Z_{1 - \alpha} - \frac{ a_*   }{\sqrt{\sum_{i=1}^I \bar{\bar \pi}_i (1 - \bar{\bar \pi}_i)  } }  \middle| A_0 = a_* \right) \\
    & = 1 - \mathbb{P} \left(\frac{T - a_*  - \sum_{i=1}^I \bar{\bar \pi}_i }{\sqrt{\sum_{i=1}^I \bar{\bar \pi}_i (1 - \bar{\bar \pi}_i)  } } \leq  Z_{1 - \alpha} - \frac{ a_*   }{\sqrt{\sum_{i=1}^I \bar{\bar \pi}_i (1 - \bar{\bar \pi}_i)  } }  \middle| A_0 = a_* \right) \\
    & = 1 - \Phi \left(Z_{1 - \alpha} - \frac{ a_*   }{\sqrt{\sum_{i=1}^I \bar{\bar \pi}_i (1 - \bar{\bar \pi}_i)  } }  \right)
\end{align*}
where $Z_{1-\alpha}$ is the $Z$-score at $1 - \alpha$ level. Above is the power calculation without considering additional case description information. When case description information is used,  assume that the null hypothesis will be rejected at the significance level $\alpha$ again, and the alternative hypothesis is true and $A_0^{(k)} = a_*^{(k)}$ in $\mathcal{C}_k$ for $k = 1, 2, \cdots,L$, and $T_k = \sum_{i \in \mathcal{C}_k}C_i$ , analogously, 

\begin{align*}
   & \qquad B(a_*^{(1)}, a_*^{(2)}, \cdots, a_*^{(L)}) \\
    & = \mathbb{P} \left( \sum_{k=1}^L \frac{T_k - a_*^{(k)} + a_*^{(k)} - \sum_{i \in \mathcal{C}_k } \bar{\bar \pi}_i  }{ \sqrt{\sum_{i \in \mathcal{C}_k} \bar{\bar \pi}_i (1 - \bar{\bar \pi}_i)  }   } \middle/ \sqrt{L} > Z_{1 - \alpha} \middle| A_0^{(1)} = a_*^{(1)}, A_0^{(2)} = a_*^{(2)} \cdots A_0^{(L)} =  a_*^{(L)}   \right) \\
        & = 1 - \mathbb{P} \left( \sum_{k=1}^L  \frac{T_k - a_*  - \sum_{i \in \mathcal{C}_k} \bar{\bar \pi}_i }{\sqrt{\sum_{i \in \mathcal{C}_k} \bar{\bar \pi}_i (1 - \bar{\bar \pi}_i)  } } \leq  Z_{1 - \alpha} - \sum_{k=1}^L \frac{ a_*^{(k)} }{ \sqrt{\sum_{i \in \mathcal{C}_k} \bar{\bar \pi}_i (1 - \bar{\bar \pi}_i)  } }  \middle| A_0^{(1)} = a_*^{(1)}, \cdots A_0^{(L)} =  a_*^{(L)} \right) \\
    & = 1 - \Phi \left( Z_{1 - \alpha} - \sum_{k=1}^L \frac{ a_*^{(k)} }{ \sqrt{\sum_{i \in \mathcal{C}_k} \bar{\bar \pi}_i (1 - \bar{\bar \pi}_i)  }   } \middle/ \sqrt{L}   \right)
\end{align*}
by Stouffer's Z-score method, and 
\begin{align*}
  & \qquad    B(a_*^{(1)}, a_*^{(2)}, \cdots, a_*^{(L)}) \\
      & = \mathbb{P} \left( \sum_{k=1}^L w_k \frac{T_k - a_*^{(k)} + a_*^{(k)} - \sum_{i \in \mathcal{C}_k } \bar{\bar \pi}_i  }{ \sqrt{\sum_{i \in \mathcal{C}_k} \bar{\bar \pi}_i (1 - \bar{\bar \pi}_i)  }   } \middle/ \sqrt{\sum_{k=1}^L w_k^2} > Z_{1 - \alpha} \middle| A_0^{(1)} = a_*^{(1)}, A_0^{(2)} = a_*^{(2)} \cdots A_0^{(L)} =  a_*^{(L)}   \right) \\
         &= 1 - \Phi \left( Z_{1 - \alpha} - \sum_{k=1}^L  \frac{ w_k a_*^{(k)} }{ \sqrt{\sum_{i \in \mathcal{C}_k} \bar{\bar \pi}_i (1 - \bar{\bar \pi}_i)  }   } \middle/ \sqrt{\sum_{k=1}^L w_k^2}   \right)
\end{align*}
where $w_k = \sqrt{|\mathcal{C}_k|}$ by weighted Stouffer's Z-score method. To investigate when using case description information will increase the statistical power, we propose the following condition, that is, if 
the sum of attributable effects obtained from these subtypes is greater than the attributable effects obtained from overall population without using case desription information, then using the case definition information leads higher statistical power of sensitivity analysis. To see this, let $k^* = \argmax_k \sum_{i \in \mathcal{C}_k}\bar{\bar \pi}_i (1 - \bar{\bar \pi}_i)  $, we observe that
\begin{align*}
    \sum_{k=1}^L \frac{ a_*^{(k)} }{ \sqrt{\sum_{i \in \mathcal{C}_k} \bar{\bar \pi}_i (1 - \bar{\bar \pi}_i)  }} \geq  \frac{\sum_{k=1}^L  a_*^{(k)} }{  \sqrt{ \sum_{i \in \mathcal{C}_{k^*}} \bar{\bar \pi}_i (1 - \bar{\bar \pi}_i)  }} \geq  \frac{ a_*   }{\sqrt{\sum_{i=1}^I \bar{\bar \pi}_i (1 - \bar{\bar \pi}_i)  } }
\end{align*}
if $\sum_{k=1}^L  a_*^{(k)} \geq a_*$. Hence, $B(a_*^{(1)}, a_*^{(2)}, \cdots, a_*^{(L)}) \geq B(a_*)$ which implies a higher statistical power of sensitivity analysis.

\section*{Appendix VI: Summary Data of Case Study}

Table \ref{table:7} provides summary data for the case study.

\begin{table}
\hspace{-2cm}
\caption{\label{table:7} Top: Frequency of exposure to hormone sensitive invasive breast cancer with 3,154 pairs, the odds ratio is 2.00 (95\% CI: $[1.39, 2.88]$). Middle: Frequency of exposure to hormone insensitive invasive breast cancer with 892 pairs, the odds ratio is 0.71 (95\% CI: $[0.37, 1.38]$).  Bottom: Frequency of exposure to invasive breast cancer with 4,046 pairs, the odds ratio is 1.58 (95\% CI: $[1.16, 2.16]$).  }
\centering
\begin{tabular}{|| c c c ||}
\hline
 Outcome & Non-invasive breast cancer ($+$)  & Non-invasive breast cancer ($-$)      \\
 \hline\hline
Invasive breast cancer ($+$) & $1$ & $86$    \\ 
Invasive breast cancer ($-$) & $43$ & $3024$   \\  
 \hline
 \hline
 Outcome & Non-invasive breast cancer ($+$)  & Non-invasive breast cancer ($-$)      \\
 \hline\hline
Invasive breast cancer ($+$) & $1$ & $15$    \\ 
Invasive breast cancer ($-$) & $21$ & $855$   \\  
 \hline
 \hline
 Outcome & Non-invasive breast cancer ($+$)  & Non-invasive breast cancer ($-$)      \\
 \hline\hline
Invasive breast cancer ($+$) & $2$ & $101$    \\ 
Invasive breast cancer ($-$) & $64$ & $3879$   \\  
 \hline
\end{tabular}
\end{table}

\section*{Appendix VII: Additional Simulations}

Table \ref{table:11} and \ref{table:12} provide additional simulation results. 


\begin{table}
\caption{\label{table:11} Power simulations include directly merged method, Stouffer's $Z$-score method, Fisher's method, truncated method with $\tilde{\alpha} = 0.05$, and Bonferronni method. Each situation is sampled $200$ times. There are $500$ data points in each dataset with average treatment effects $\delta_1$ and $\delta_2$. $\Theta = 1.1$.  }
\centering
\scalebox{0.7}{
\begin{tabular}{||c c c c c c c c c c||} 
 \hline
 \multicolumn{10}{||c||}{Equal effects: $\delta_1 = 0.2, \delta_2 = 0.2$}   \\
 \hline \hline
   & $\Gamma = 1.0$  & $\Gamma = 1.5$  & $\Gamma = 2.0$  & $\Gamma = 2.5$  & $\Gamma = 3$ & $\Gamma = 3.5$ & $\Gamma = 4$ & $\Gamma = 4.5$ & $\Gamma = 5$  \\ 
 \hline
 Merged & 1.000 & 1.000 & 1.000 & 0.970 & 0.665 & 0.320 & 0.105 & 0.020 & 0.010 \\
 Stouffer & 1.000 &  1.000 & 1.000 & 0.960 & 0.660 & 0.300 & 0.105 & 0.030 & 0.010 \\ 
 Fisher & 1.000 &  1.000 & 1.000 & 0.950 & 0.610 & 0.280 & 0.100 & 0.030 & 0.010\\ 
 Truncated $\tilde{\alpha} = 0.05$ & 1.000 &  1.000 & 0.995 & 0.875 & 0.485 & 0.235 & 0.070 & 0.040 & 0.010 \\
 Bonferroni & 1.000 &  1.000 & 0.995 & 0.835 & 0.450 & 0.230 & 0.070 & 0.040 & 0.010\\
 \hline \hline

 \multicolumn{10}{||c||}{Slightly unequal effects: $\delta_1 = 0.3, \delta_2 = 0.2$} \\
 \hline \hline
   & $\Gamma = 1.0$  & $\Gamma = 1.5$  & $\Gamma = 2.0$  & $\Gamma = 2.5$  & $\Gamma = 3$ & $\Gamma = 3.5$ & $\Gamma = 4$ & $\Gamma = 4.5$ & $\Gamma = 5$  \\ 
 \hline
 Merged & 1.000 &  1.000 & 1.000 & 1.000 & 0.920 & 0.645 & 0.395 & 0.205 & 0.080 \\
 Stouffer & 1.000 &  1.000 & 1.000 & 0.995 & 0.945 & 0.720 & 0.420 & 0.210 & 0.110 \\ 
 Fisher & 1.000 &  1.000 &  1.000 & 0.995 & 0.915 & 0.665 & 0.370 & 0.185 & 0.100\\ 
 Truncated $\tilde{\alpha} = 0.05$ & 1.000 &  1.000 &  1.000 & 0.975 & 0.855 & 0.600 & 0.330 & 0.155 & 0.085 \\
 Bonferroni & 1.000 &  1.000 &  1.000 & 0.975 & 0.845 & 0.575 & 0.325 & 0.150 & 0.085\\
 \hline \hline
 \multicolumn{10}{||c||}{Unequal effects: $\delta_1 = 0.6, \delta_2 = 0.2$} \\
 \hline \hline
   & $\Gamma = 1.0$  & $\Gamma = 1.5$  & $\Gamma = 2.0$  & $\Gamma = 2.5$  & $\Gamma = 3$ & $\Gamma = 3.5$ & $\Gamma = 4$ & $\Gamma = 4.5$ & $\Gamma = 5$  \\ 
 \hline
 Merged & 1.000 &  1.000 &    1.000 &  1.000 &  1.000 & 0.985 & 0.945 & 0.715 & 0.565 \\
 Stouffer & 1.000 &  1.000 &    1.000 &  1.000 &    1.000 & 1.000 & 0.980 & 0.895 & 0.660 \\ 
 Fisher & 1.000 &  1.000 &    1.000 &  1.000 &    1.000 & 1.000 & 1.000 & 0.995 & 0.970 \\ 
 Truncated $\tilde{\alpha} = 0.05$ & 1.000 &  1.000 &  1.000 &  1.000 &  1.000 & 1.000 & 1.000 & 1.000 & 0.990 \\ 
 Bonferroni & 1.000 &  1.000 &  1.000 & 1.000 & 1.000 & 1.000 & 1.000 & 1.000 & 0.990 \\
 \hline \hline
 \multicolumn{10}{||c||}{Effect only in one dataset: $\delta_1 = 0.6, \delta_2 = 0.0$} \\
 \hline \hline
   & $\Gamma = 1.0$  & $\Gamma = 1.5$  & $\Gamma = 2.0$  & $\Gamma = 2.5$  & $\Gamma = 3$ & $\Gamma = 3.5$ & $\Gamma = 4$ & $\Gamma = 4.5$ & $\Gamma = 5$  \\ 
 \hline
 Merged & 1.000 &  1.000 &  1.000 &  1.000 &    1.000 & 0.955 & 0.840 & 0.600 & 0.360 \\
   Stouffer & 1.000 &  1.000 &  1.000 &  1.000 & 1.000 & 0.990 & 0.945 & 0.770 & 0.550 \\ 
 Fisher & 1.000 &  1.000 &  1.000 &  1.000 & 1.000 & 1.000 & 1.000 & 0.990 & 0.970\\ 
 Truncated $\tilde{\alpha} = 0.05$ & 1.000 &  1.000 &  1.000 &  1.000 &    1.000 & 1.000 & 1.000 & 1.000 & 0.990 \\  
 Bonferroni & 1.000 & 1.000 & 1.000 & 1.000 & 1.000 & 1.000 & 1.000 & 1.000 & 0.990\\
 \hline 
\end{tabular}}
\end{table}

\begin{table}
\caption{\label{table:12} Power simulations include directly merged method, Stouffer's $Z$-score method, Fisher's method, truncated method with $\tilde{\alpha} = 0.05$, and Bonferronni method. Each situation is sampled $200$ times. There are $500$ data points in each dataset with average treatment effects $\delta_1$ and $\delta_2$. $\Theta = 1.2$.   }
\centering
\scalebox{0.7}{
\begin{tabular}{||c c c c c c c c c c||} 
 \hline
 \multicolumn{10}{||c||}{Equal effects: $\delta_1 = 0.2, \delta_2 = 0.2$}   \\
 \hline \hline
   & $\Gamma = 1.0$  & $\Gamma = 1.5$  & $\Gamma = 2.0$  & $\Gamma = 2.5$  & $\Gamma = 3$ & $\Gamma = 3.5$ & $\Gamma = 4$ & $\Gamma = 4.5$ & $\Gamma = 5$  \\ 
 \hline
 Merged & 1.000 &  1.000 & 0.990 & 0.890 & 0.545 & 0.210 & 0.055 & 0.020 &  0.000  \\
 Stouffer & 1.000 &  1.000 & 0.990 & 0.895 & 0.540 & 0.200 & 0.045 & 0.020 &    0.000 \\ 
 Fisher & 1.000 &  1.000 & 0.995 & 0.855 & 0.505 & 0.185 & 0.035 & 0.015 &  0.000 \\ 
 Truncated $\tilde{\alpha} = 0.05$ & 1.000 &  1.000 & 0.995 & 0.770 & 0.415 & 0.140 & 0.045 & 0.005 &  0.000 \\
 Bonferroni  & 1.000 & 1.000 & 0.985 & 0.745 & 0.395 & 0.130 & 0.040 & 0.005 & 0.000\\
 \hline \hline

 \multicolumn{10}{||c||}{Slightly unequal effects: $\delta_1 = 0.4, \delta_2 = 0.2$} \\
 \hline \hline
   & $\Gamma = 1.0$  & $\Gamma = 1.5$  & $\Gamma = 2.0$  & $\Gamma = 2.5$  & $\Gamma = 3$ & $\Gamma = 3.5$ & $\Gamma = 4$ & $\Gamma = 4.5$ & $\Gamma = 5$  \\ 
 \hline
 Merged & 1.000 &  1.000 & 1.000 & 0.980 & 0.775 & 0.480 & 0.175 & 0.075 & 0.030 \\
 Stouffer & 1.000 &  1.000 & 1.000 & 0.990 & 0.795 & 0.510 & 0.195 & 0.105 & 0.030 \\ 
 Fisher & 1.000 &  1.000 & 1.000 & 0.955 & 0.795 & 0.480 & 0.185 & 0.070 & 0.030\\ 
 Truncated $\tilde{\alpha} = 0.05$ & 1.000& 1.000 & 0.995 & 0.900 & 0.680 & 0.400 & 0.195 & 0.055 & 0.030 \\ 
 Bonferroni & 1.000 & 1.000 & 0.995 & 0.890 & 0.660 & 0.400 & 0.195 & 0.055 & 0.030\\
 \hline \hline
 \multicolumn{10}{||c||}{Unequal effects: $\delta_1 = 0.6, \delta_2 = 0.2$} \\
 \hline \hline
   & $\Gamma = 1.0$  & $\Gamma = 1.5$  & $\Gamma = 2.0$  & $\Gamma = 2.5$  & $\Gamma = 3$ & $\Gamma = 3.5$ & $\Gamma = 4$ & $\Gamma = 4.5$ & $\Gamma = 5$  \\ 
 \hline
 Merged & 1.000 &  1.000 &  1.000 &  1.000 &  1.000 & 0.990 & 0.875 & 0.685 & 0.440 \\
 Stouffer & 1.000 &  1.000 & 1.000 &  1.000 &  1.000 & 1.000 & 0.960 & 0.820 & 0.630 \\ 
 Fisher & 1.000 &  1.000 & 1.000 &  1.000 &  1.000 & 1.000 & 1.000 & 0.990 & 0.965 \\ 
 Truncated $\tilde{\alpha} = 0.05$ & 1.000 &  1.000 &  1.000 & 1.000 & 1.000 & 1.000 & 1.000 & 1.000 & 0.995 \\ 
 Bonferroni & 1.000 &  1.000 &  1.000 & 1.000 & 1.000 & 1.000 & 1.000 & 1.000 & 0.995 \\
 \hline \hline
 \multicolumn{10}{||c||}{Effect only in one dataset: $\delta_1 = 0.6, \delta_2 = 0.0$} \\
 \hline \hline
   & $\Gamma = 1.0$  & $\Gamma = 1.5$  & $\Gamma = 2.0$  & $\Gamma = 2.5$  & $\Gamma = 3$ & $\Gamma = 3.5$ & $\Gamma = 4$ & $\Gamma = 4.5$ & $\Gamma = 5$  \\ 
 \hline
 Merged & 1.000 &  1.000 &  1.000 &  1.000 &  1.000 & 0.905 & 0.640 & 0.425 & 0.180 \\
   Stouffer & 1.000 &  1.000 & 1.000 &  1.000 &  1.000 & 0.955 & 0.840 & 0.590 & 0.330 \\ 
 Fisher & 1.000 &  1.000 & 1.000 &  1.000 &   1.000 & 1.000 & 1.000 & 0.995 & 0.940\\ 
 Truncated $\tilde{\alpha} = 0.05$ & 1.000 & 1.000 & 1.000 & 1.000 & 1.000 &  1.000 &  1.000 &  1.000 &  0.995 \\  
 Bonferroni & 1.000 & 1.000 & 1.000 & 1.000 & 1.000 & 1.000 & 1.000 & 1.000 & 0.995\\
 \hline 
\end{tabular}}
\end{table}

\end{document}